\documentclass[journal,twoside,web]{ieeecolor}
\usepackage{generic}
\usepackage{graphicx}      
\usepackage{tabulary}
\usepackage{amsfonts}
\usepackage{amsmath}
\usepackage{amssymb}
\let\oldsquare\square
\usepackage{pgf,tikz}
\usepackage{tabularx}
\usepackage{float}
\usepackage{cite}
\usepackage{stfloats}
\usepackage{mathtools}
\usepackage{mathabx}
\usepackage[ruled,vlined,linesnumbered,resetcount]{algorithm2e}
\usepackage{setspace}
\usepackage{booktabs}

\let\labelindent\relax
\usepackage{enumitem} 
\renewcommand{\theenumi}{\roman{enumi}}
\usepackage{bbm} 
\usepackage{bm} 
\usepackage{textcomp}

\newtheorem{theorem}{Theorem}[section]
\newtheorem{proposition}[theorem]{Proposition}
\newtheorem{lemma}[theorem]{Lemma}

\newtheorem{remark}[theorem]{Remark}

\newtheorem{problem}{Problem}

\newtheorem{assumption}{Assumption}

\newcommand{\col}{{\rm col\;}}
\newcommand{\diag}{{\rm diag\;}}

\newcommand{\maxv}{{\rm vmax}}
\newcommand{\minv}{{\rm vmin}}
\newcommand*{\horzbar}{\rule[.5ex]{2.5ex}{0.5pt}}

\makeatletter
\renewcommand*\env@matrix[1][\arraystretch]{%
	\edef\arraystretch{#1}%
	\hskip -\arraycolsep
	\let\@ifnextchar\new@ifnextchar
	\array{*\c@MaxMatrixCols c}}
\makeatother

\newcommand{\Jc}{\mathcal{J}}
\newcommand{\until}[1]{\{1,\dots,#1\}}
\newcommand{\longthmtitle}[1]{\mbox{} \emph{(#1):}}

\renewcommand{\baselinestretch}{.99}
\allowdisplaybreaks

\def\qed{ \rule{.1in}{.1in}}
\def\proof{\noindent\hspace{2em}{\itshape Proof: }}

\def\BibTeX{{\rm B\kern-.05em{\sc i\kern-.025em b}\kern-.08em
		T\kern-.1667em\lower.7ex\hbox{E}\kern-.125emX}}
\begin{document}




\title{Efficient Reconstruction of Neural Mass Dynamics Modeled by
  Linear-Threshold Networks}


\author{Xuan Wang, \IEEEmembership{Member, IEEE}, Jorge Cort{\'e}s, \IEEEmembership{Fellow, IEEE} 
\thanks{A preliminary version of this work appeared as \cite{XW-JC:21-cdc} at the IEEE
    Conference on Decision and Control. 
    X. Wang is with the Department of Electrical and Computer
    Engineering, George Mason University, Fairfax,
    xwang64@gmu.edu. This work started when he was a postdoctoral
    researcher at University of California, San Diego.  J. Cort{\'e}s
    is with the Department of Mechanical and Aerospace Engineering,
    University of California, San Diego, cortes@ucsd.edu}} %

\maketitle

\begin{abstract}
  This paper studies the data-driven reconstruction of firing rate
  dynamics of brain activity described by linear-threshold network
  models.  Identifying the system parameters directly leads to a large
  number of variables and a highly non-convex objective
  function. Instead, our approach introduces a novel reformulation
  that incorporates biological organizational features and turns the
  identification problem into a scalar variable optimization of a
  discontinuous, non-convex objective function.  We prove that the
  minimizer of the objective function is unique and establish that the
  solution of the optimization problem leads to the identification of
  all the desired system parameters. These results are the basis to
  introduce an algorithm to find the optimizer by searching the
  different regions corresponding to the domain of definition of the
  objective function.  To deal with measurement noise in sampled data,
  we propose a modification of the original algorithm whose
  identification error is linearly bounded by the magnitude of the
  measurement noise.  We demonstrate the effectiveness of the proposed
  algorithms through simulations on synthetic and experimental data.
\end{abstract}

\begin{IEEEkeywords}
Dynamics Reconstruction; Linear-Threshold Networks ; Neural Mass Models; System identification.
\end{IEEEkeywords}

\section{Introduction}
The realization of complex brain functions relies critically on the
interaction among billions of neuron cells.  Such brain activity can
be modeled and analyzed in a quantitative way using neural mass
models, which describe the evolution of the firing rate of neurons
(e.g., number of spikes per second) and have good trial-to-trial
reproducibility~\cite{WG-AKK-HM-AVMH:97} and accessibility.  The
firing activity of single neurons can be recorded through
cell-attached recording techniques \cite{PA-RF-IL-AM:12}, and the
combined firing activity of populations of neurons can be measured by
electrocorticography (ECoG) \cite{NEC-AS-AK:06}, allowing researchers
to analyze brain systems at different scales. In computational
neuroscience, the meso-scale\footnote{A meso-scale model sits between
  the micro-scale, which takes single neurons as entities, and the
  macro-scale which takes brain regions as
  entities.}\cite{COB-DSB-VMP:18} brain neuronal interactions can be
described by network models \cite{DSB-OS:17,MR-OS:10}, where each node
of the network represents a population of adjacent neurons; the state
of the node is governed by local dynamics characterizing the neurons'
average firing rate; and the edges are defined by the connected neuron
populations whose firing rates are interactive.  Such models are
structurally consistent with brain neuronal activities, both
physiologically and anatomically
\cite{MS-ARM-VJ-GD-PR:18}. Nevertheless, determining their edge
weights is usually challenging because of the difficulty of measuring
and quantifying the strength of neuronal interactions.  Motivated by
this, our research focuses on using sampled data to reconstruct the
firing rate dynamics of brain neural networks, with the ultimate goal
of enabling prediction and control of such models.  Since data
collection about neural systems is subject to uncertainty in their
behavior, including firing rates, such reconstruction needs to take
into account the impact of measurement noise and modeling error.

\subsubsection*{Literature review}

In computational neuroscience, firing rates and blood-oxygen-level
dependence (BOLD)\cite{TM-GB-DSB-FP:19} are two common approaches to
quantifying brain neural activity. BOLD signals can be collected by
functional magnetic resonance imaging (fMRI) scans, which have
relatively low spatial ($\sim$ 1$mm^3$) and temporal ($\sim$ 2$s$)
resolutions\cite{PM-RC-JBM:12}. In contrast, although the collection
of firing rates is more challenging and invasive, requiring the
insertion of electrodes through surgery, the spatial ($\sim$
30$\mu m^3$) and temporal ($\sim$ 0.2$ms$) resolutions
\cite{ES-BS-IV-HMP-IN-MS-IS:15} of its measurements are significantly
better.
Recently~\cite{COB-DSB-VMP:18,TM-GB-DSB-FP:19}, BOLD has been
successfully used at the meso scale level of network modeling to
understand interactive brain activities. Likewise, the better
resolutions offered by firing rates allow to build more precise
network models, and based on these, develop prediction or control
schemes to study brain behavior from a dynamical perspective.  Towards
this end, the results in \cite{DS-SVS:14,DE-DS-SVS:15} use linear
network models to describe firing rate dynamics by assuming the
neurons' local dynamics can be linearized around their fixed points.
However, such simplification ignores two important properties of
firing rates, i.e., the values are non-negative and are subject to
saturation constraints.  In this paper,
following~\cite{PD-LFA:01,EN-JC:21-tacI,FC-AA-FP-JC:21-csl}, we employ
a linear-threshold network model to describe the dynamical behavior of
firing rates that takes into account these properties.

A key step in our work is to determine the parameters of the
meso-scale network model, including time constants, edge weights, and
threshold bounds.  This process is closely related to system
identification (SysID)\cite{ZH-ZW:13}. Given an unknown system, SysID
aims to learn the system parameters from its input-output data. With
an abundant literature in this field, powerful methods have been
proposed for the identification of linear systems~\cite{KJA-PE:71}.
However, given the inherent complexity and variety in model structures
of nonlinear systems, unified SysID approaches for them usually lack
provable guarantees on the accuracy of the identified parameters, and
their computational complexity is
significant~\cite{XH-RJM-SC-CJH-KL-GWI:08}.
The identification of linear-threshold network models addressed here
has parallelisms with the determination of weights in the training of
neural networks with the rectified linear unit (RELU) activation
function~\cite{AB-RJ-DPW-19} in the machine learning
literature~\cite{FA-EDS-VM:93,GH:03}.  However, we note that the
research goals and methods are fundamentally different, mostly
stemming from the connection (or lack thereof) to actual physical
processes associated to the identified network model.  Here, we seek to
reconstruct the dynamical behavior of an actual physical system, whose
nodes' states evolve with time, corresponding to their current states
and the system input.  Instead, when training neural networks, the
weights have no physical or dynamical relevance and the static
network model seeks to establish a virtual input-output mapping.


\subsubsection*{Statement of contributions}
We study the reconstruction of firing rate dynamics in a
linear-threshold network model based on discrete-time data samples.
We start by noting that the identification of all model parameters
gives rise to a highly non-convex and non-smooth problem with a large
number of variables. In turn, this means that: (a) solving the problem
directly is computationally expensive; and (b) solutions obtained from
local minimizers are not robust against measurement noise. In order to
address these issues, the contributions of the paper are
two-fold. First, we introduce a new approach with lower computational
complexity for parameter identification. Based on a reformulation of
the linear-threshold model, the proposed approach optimizes a
discontinuous and non-convex function which is only a function of a
scalar variable and is piecewise smooth.  This reformulation can also
take into account Dale's law, which arises from the physiology of
neurotransmission and introduces sign constraints on the model's edge
weights.  We show through analysis that the new objective function has
a unique minimizer, under appropriate conditions on the sampled data,
and that the minimizer can be used to compute all the desired
parameters of the linear-threshold network model. This allows us to
develop an algorithm to obtain system parameters based on the scalar
optimization and analyze its computational complexity.  Our second
contribution deals with the measurement noise in sampled data. For a
general non-convex optimization problem, bounded measurement noise may
lead to unbounded changes to its solution. To avoid this, we modify
the proposed algorithm, making it robust to the impact of measurement
noise. When the sampled data involves bounded noise, our analysis
shows that the identification error of the algorithm is linearly
bounded by the magnitude of the measurement noise.  For both
algorithms, we validate their effectiveness in synthetic and
experimental data from the activity of rodents' brains executing a
selective listening task.
 
\subsubsection*{Notation}
Let ${\bf 1}_r$ denote the vector in $\mathbb{R}^r$ with all entries
equal to $1$. Let $I_r$ denote the $r\times r$ identity matrix.  We
let
$\col\{A_1,A_2,\cdots,A_r\}=\begin{bmatrix} A_1^{\top} &A_2^{\top}&
  \cdots & A_r^{\top}
\end{bmatrix}^{\top}$ be a vertical stack of matrices $A_i$ possessing
the same number of columns. 
Let $\diag\{A_1,A_2,\cdots,A_r\}$ be a block diagonal matrix with $A_i$
the $i$th diagonal block entry. Let $\maxv(x),\minv(x)\in\mathbb{R}$ 
be the component-wise maximum/minimum of vector $x$, respectively. 
Let $x[i]\in\mathbb{R}$ be the $i$th 
entry of vector $x$; correspondingly, let $M[i,j]\in\mathbb{R}$ be the entry
of matrix $M$ on its $i$th row and $j$th column.  Let $M^{\top}$ be
the transpose of a matrix $M$.  Let $|\Omega|$ be the cardinality of 
a set $\Omega$.
For $x\in\mathbb{R}$, define the
threshold function $[x]_0^s$ with $s>0$ as
\begin{align*}
  [x]_0^s=\left\{\begin{matrix*}[l]
      s & \text{for}\quad x>s ,	\\
      x & \text{for}\quad 0\le x\le s ,\\
      0 & \text{for}\quad x< 0 .
\end{matrix*}\right.  
\end{align*}
For a vector $x$, $[x]_0^s$ denotes the component-wise application of
these definitions.  For $x\in\mathbb{R}^{r}$ and $1\le i\le r$,
$x_{-i}$ denotes the vector in $\mathbb{R}^{r-1}$ obtained by removing
the $i$th entry of~$x$.

\section{Problem Formulation}
In this section, we first introduce a continuous-time firing rate
dynamical model for neuronal networks following~\cite{PD-LFA:01} and
then convert it to its discrete-time form.

Consider a network, where each node represents a population of neurons
with similar activation patterns, evolving according to
linear-threshold dynamics, for ${t}\ge0$,
\begin{align}\label{eq_cntmodel}
  \tau\dot{\bm{x}}({t})=- \bm{x}({t}) +
  \left[{W}\bm{x}({t})+B\bm{u}({t})\right]_0^{{s}},
\end{align}
Here, $\tau$ is a time constant capturing the
timescale~\cite{WG-AKK-HM-AVMH:97} of the neuronal system,
$\bm{x}\in\mathbb{R}^n$, $\bm{x}\ge 0$ is the system state,
corresponding to the firing rate of the nodes; and
$W\in\mathbb{R}^{n\times n }$ is the synaptic connectivity matrix,
characterizing the interactions (excitation or inhibition) between
different nodes. For $i \in \until{n}$, we assume $W[i,i]=0$, that is,
the nodes do not have self-loops. $\bm{u}\in \mathbb{R}^m$ and
$B\in\mathbb{R}^{n\times m}$ ($m\le n$) are the external inputs and
the associated input matrix. For each node, the stimulation it
receives from its neighboring nodes and external inputs is
non-negative and bounded by a threshold $s$, denoted by
$[\cdot]_0^s$. 

The discretization of the system~\eqref{eq_cntmodel} by the forward
Euler method with a constant step-size $\delta\ll \tau$ yields
\begin{align}\label{eq_dzmodel}
  \frac{\tau}{\delta}\left(\bm{x}^+- \bm{x}\right)=- \bm{x} +
  \left[{W}\bm{x}+{B}\bm{u}\right]_0^{{s}}.
\end{align}
Here, $\bm{x}$, $\bm{u}$ are the current system state and input, and
$\bm{x}^+$ is the system state after the interval $\delta$.  For
convenience of presentation, let
\begin{align*}
  \alpha\triangleq 1-\frac{\delta}{\tau} \in  (0,1), \; W_D
  \triangleq \frac{\delta}{\tau}{W}, \;
  B_D\triangleq\frac{\delta}{\tau}{B}, \;
  s_D\triangleq\frac{\delta}{\tau}{s}.
\end{align*}
be the parameters of the discrete-time system. Then \eqref{eq_dzmodel} 
can be rewritten into an equivalent form as:
\begin{align}\label{eq_dstmodel} 
  \bm{x}^+=\alpha \bm{x} + \left[W_D\bm{x}+B_D\bm{u}\right]_0^{s_{_D}}.
\end{align}
We assume the system states $\bm{x}$, $\bm{x}^+$, and the system
inputs $\bm{u}$ can be sampled.
%
%
We denote the data samples by
$\bm{x}_d(k)$, $\bm{x}^+_d(k)$ and $\bm{u}_d(k)$, respectively, for
$k \in \until{T_d}$, where $T_d$ is the total number of data
sets. 

\begin{remark}\longthmtitle{Data collection}
  Note that the index $k$ in the notations $\bm{x}_d(k)$,
  $\bm{x}^+_d(k)$ and $\bm{u}_d(k)$ is simply an indicator that
  distinguishes one data sample from another. In fact, for each sample
  set, we only require that the time interval between $\bm{x}^+_d(k)$
  and $\bm{x}_d(k)$ satisfies the discretization step-size $\delta$.
  Of course, it is possible that all the sampling instances of the
  data are chosen consecutively from a system trajectory with a fixed
  interval $\delta$, which means that all the data samples are
  head-tail connected, i.e., $\bm{x}^+_d(k)$ of the former data can be
  used as the $\bm{x}_d(k)$ of the latter one. {However}, in general,
  we allow the data samples to be collected at independent time
  instances, and even from various trajectories of the same
  system. \hfill$\oldsquare$
\end{remark}

\begin{problem}\label{Prob}
  Given data samples $\bm{x}_d(k)$, $\bm{x}^+_d(k)$ and $\bm{u}_d(k)$,
  $k \in \until{T_d}$, identify the parameters $\alpha$, $W_D$, $B_D$,
  and $s_D$ of system \eqref{eq_dstmodel}.
\end{problem}

To solve this problem, one could seek to fit the
model~\eqref{eq_dstmodel} with the given data samples $\bm{x}_d(k)$,
$\bm{x}^+_d(k)$ and $\bm{u}_d(k)$.  However, due to the presence of
the (non-linear, non-convex) threshold function, such an approach
would involve a non-convex minimization problem with a large number of
variables.  Motivated by this observation, we seek to develop a more
efficient approach that exploits the specific structure
of~\eqref{eq_dstmodel}.

\section{Scalar Optimization for Parameter
  Identification}\label{SEC_KI} 
In this section, we reformulate the parameter identification as a
scalar variable optimization problem. This sets the basis for the
development of our algorithmic procedure to identify the parameters of
the firing-rate model.

\subsection{Data-based parameter identification}
For $ k \in \until{T_d}$, bringing the system inputs $\bm{u}_d(k)$ and
states $\bm{x}_d(k)$, $\bm{x}^+_d(k)$ into~\eqref{eq_dstmodel}, we
have
\begin{align}\label{eq_dataF}
  \bm{x}^+_d(k)-\alpha\bm{x}_d(k)= \left[H \bm{p}_d(k)\right]_0^{s_{_D}},
\end{align}
where $\bm{p}_d(k)=\col\{\bm{x}_d(k),\bm{u}_d(k)\}$ and
\begin{align}\label{eq_defzHP}
  H&=\begin{bmatrix} W_D&B_D
  \end{bmatrix}=\begin{bmatrix}[1.2]
    \horzbar 	& h_1^{\top} & \horzbar \\
    \horzbar 	& h_2^{\top} & \horzbar \\
    &	\vdots &\\[.2em]
    \horzbar& 	h_n^{\top} & \horzbar
  \end{bmatrix}\in\mathbb{R}^{n\times (n+m)} .
\end{align}
Note that in \eqref{eq_defzHP}, since the diagonal entries of $W_D$
are zero, i.e., $h_i[i]=0$, not all the entries of $H$ are variables
that need to be parameterized for identification. To characterize
this, for $i \in \until{n}$, define
$\bar{h}_i= (h_i)_{-i} \in\mathbb{R}^{n+m-1}$, which removes the $i$th
entry from $h_i$.  Correspondingly, let
$\bar{\bm{p}}_i (k)=(\bm{p}_d(k))_{-i}$.
Let $\bm{h} =
\col\{\bar{h}_1,\bar{h}_2,\dots,\bar{h}_n\}\in\mathbb{R}^{n(n+m-1)}$
and $\bm{P}_d(k) = \diag\{\bar{\bm{p}}_1^{\top}(k),
\bar{\bm{p}}_2^{\top}(k),\dots,\bar{\bm{p}}_n^{\top}(k)\}
\in\mathbb{R}^{n \times n(n+m-1)}$.
%
%
Then, one can write
\begin{align}\label{eq_dataF2}  
  H \bm{p}_d(k)&=\begin{bmatrix}[1.2]
    h_1^{\top}\bm{p}_d(k)\\
    h_2^{\top}\bm{p}_d(k)\\
    \vdots\\
    h_n^{\top}\bm{p}_d(k)
  \end{bmatrix}=\begin{bmatrix}[1.2]
    ~\bar{h}_1^{\top}\bar{\bm{p}}_1(k)\\
    ~\bar{h}_2^{\top}\bar{\bm{p}}_2(k)\\
    ~\vdots\\
    ~\bar{h}_n^{\top}\bar{\bm{p}}_n(k)
  \end{bmatrix}=\begin{bmatrix}[1.2]
    \bar{\bm{p}}_1^{\top}(k)\bar{h}_1\\
    \bar{\bm{p}}_2^{\top}(k)\bar{h}_2\\
    \vdots\\
    \bar{\bm{p}}_n^{\top}(k)\bar{h}_n
  \end{bmatrix}\nonumber\\
               &=\bm{P}_d(k)\bm{h} ,
\end{align}
where the second equality holds because $h_i[i]=0$. All entries in
$\bm{h}$ are variables to be identified. To proceed, define compact
vectors/matrices:
\begin{align}\label{eq_defXXP} 
  \mathcal{X}=\begin{bmatrix}[1.2]
    \bm{x}_d(1)\\\bm{x}_d(2)\\\vdots\\\bm{x}_d({T_d})
  \end{bmatrix}, 
  \;
  \mathcal{X}^+=\begin{bmatrix}[1.2]
    \bm{x}^+_d(1)\\\bm{x}^+_d(2)\\\vdots\\\bm{x}^+_d({T_d})
  \end{bmatrix},
  \;
  \mathcal{P}=\begin{bmatrix}[1.2]
    \bm{P}_d(1)\\\bm{P}_d(2)\\\vdots\\\bm{P}_d({T_d})
  \end{bmatrix} ,
\end{align}
such that $\mathcal{X}\in\mathbb{R}^{nT_d}$,
$\mathcal{X}^+\in\mathbb{R}^{nT_d}$, and
$\mathcal{P}\in\mathbb{R}^{nT_d\times n(n+m-1)}$.
Then,~\eqref{eq_dataF} reads
\begin{align}\label{eq_compactdataeq}
	\mathcal{X}^+-\alpha \mathcal{X}=\left[\mathcal{P}\bm{h}\right]_0^{s_{_D}}
\end{align}
Now, given variables $v_i\ge0$ to be determined, let
\begin{align}\label{eq_f_2}
  f(\!\mathcal{X}^+\!\!\!-\!\alpha \mathcal{X})[i]\!=\!
  \begin{cases}
    \!\phantom{-}v_i & \text{if } \left(\!\mathcal{X}^+\!\!-\!\alpha
      \mathcal{X}\right)[i] \!=\! \maxv(\mathcal{X}^+\!\!-\!\alpha
    \mathcal{X}) ,
    \\
    \!-v_i & \text{if }\left(\mathcal{X}^+\!\!-\!\alpha
      \mathcal{X}\right)[i]= 0 ,
    \\
    \!\phantom{-}0 & \text{otherwise} ,
  \end{cases}
\end{align}
for $i \in \until{nT_d}$. Note that, with the right choice of $v_i$'s,
one can decompose
$\mathcal{P}\bm{h}=\left[\mathcal{P}\bm{h}\right]_0^{s_{_D}}+f(\mathcal{X}^+-\alpha
\mathcal{X})$, i.e., the role of $f(\mathcal{X}^+-\alpha \mathcal{X})$
is to \textit{compensate} for the parts of $\mathcal{P}\bm{h}$ that
are truncated by the threshold $[\cdot]_0^{s_{_D}}$.
Equation~\eqref{eq_compactdataeq} can then be written as
\begin{align}\label{eq_dataeqf}
  \mathcal{X}^+-\alpha
  \mathcal{X}-\mathcal{P}\bm{h}+f(\mathcal{X}^+-\alpha \mathcal{X})=0 .
\end{align}
To further simplify the non-linear mapping
$f(\mathcal{X}^+-\alpha \mathcal{X})$, we rewrite
\begin{align}\label{eq_defCv}
  f(\mathcal{X}^+-\alpha \mathcal{X})=C(\alpha)v, \qquad v\ge0
\end{align}
where for any fixed $\alpha$,
$C(\alpha)\in\mathbb{R}^{nT_d\times d(\alpha)}$ is a matrix that can
be constructed using~\eqref{eq_f_2} by the following two-step
procedure:
\begin{enumerate}
\item Define a diagonal matrix $E(\alpha)\in\mathbb{R}^{nT_d\times
    nT_d}$ such that for all $i \in \until{nT_d}$,
  \begin{align}\label{eq_defcE}
    &E(\alpha)[i,i] =\!
      \begin{cases}
        \!\phantom{-}1 & \text{if } \left(\!\mathcal{X}^+\!\!-\!\alpha
          \mathcal{X}\right)[i] \!=\!  \maxv(\mathcal{X}^+\!\!-\!\alpha
        \mathcal{X}) ,
        \\
        \!-1 & \text{if }\left(\mathcal{X}^+\!\!-\!\alpha
          \mathcal{X}\right)[i]= 0 ,
        \\
        \!\phantom{-}0 & \text{otherwise} ;
      \end{cases}
  \end{align}
\item Construct $C(\alpha)$ by removing all zero columns in
  $E(\alpha)$.
\end{enumerate}
Note that the number of columns of $C(\alpha)$, denoted by $d(\alpha)$,
is dependent on $\alpha$. This matrix has the following properties
\begin{align}\label{eq_CTC}
  C(\alpha)^{\top}\!C(\alpha)\!=\!I_{d(\alpha)} \quad \text{and} \quad
  C(\alpha)C(\alpha)^{\top}\!\!=\!E(\alpha)^2.
\end{align}
Looking at the expression~\eqref{eq_defCv} and the definition
in~\eqref{eq_f_2}, we see that the vector $v\in\mathbb{R}^{d(\alpha)}$
encodes the magnitudes of the components
of~$f(\mathcal{X}^+-\alpha \mathcal{X})$ whereas the matrix
$C(\alpha)$ encodes the corresponding signs.  The following result is
an immediate consequence of these definitions.

\begin{lemma}\longthmtitle{Matrices $E(\alpha)$ and $C(\alpha)$
    are piecewise constant}\label{LM_choiceEC}
  Given vectors $\mathcal{X}^+, \mathcal{X}\ge0$, the matrices
  $E(\alpha)$  and $C(\alpha)$ are piecewise constant functions of $\alpha$.
\end{lemma}

Note that the structure of $C(\alpha)$ and the value of $v$ depend
nonlinearly on the choice of~$\alpha$. Substituting~\eqref{eq_defCv}
into \eqref{eq_dataeqf},
\begin{align}\label{eq_eqcv}
  \mathcal{X}^+-\alpha\mathcal{X}+C(\alpha)v-\mathcal{P}\bm{h}=0,
  \qquad v\ge0 .
\end{align}
To find $\alpha$, $\bm{h}$, and $v$ that satisfy \eqref{eq_eqcv}, we
can consider them as the critical points of the following objective
function
\begin{align}\label{eq_obj2}
  \mathcal{J}_0(\alpha, v, \bm{h})=\frac{1}{2}\|\mathcal{X}^+\!-\!\alpha
  \mathcal{X}\!+\!C(\alpha)v\!-\!\mathcal{P}\bm{h}\|_2^2
\end{align}
%
%
The minimization of~\eqref{eq_obj2} is subject to the constraint $
v\ge0$.  
Thus, letting
\begin{align}\label{eq_defQxi}
  \mathcal{Q}(\alpha)
  =
  \begin{bmatrix}
    C(\alpha)
    &
    -\mathcal{P}
  \end{bmatrix},
      \;
      \xi=
      \begin{bmatrix}[1.2]
        v
        \\
        \bm{h}
      \end{bmatrix},\;
  S=
  \begin{bmatrix}
    -I_{d(\alpha)}
    & \bm{0}
  \end{bmatrix},
\end{align}
the optimization problem \eqref{eq_obj2} takes the form
\begin{align}\label{eq_obj2b}
  \begin{aligned}
    \min&
    \qquad\mathcal{J}_0(\alpha,\xi)=\frac{1}{2}\|\mathcal{X}^+-\alpha
    \mathcal{X}+ \mathcal{Q}(\alpha)\xi
    \|_2^2.\\
    \text{s.t.}&\qquad S\xi\le0
  \end{aligned}
\end{align}

\begin{remark}\longthmtitle{Incorporating Dale's law}\label{Rm_Dale}
  According to Dale's law, a neuron performs the same chemical action
  at all of its synaptic connections to others, regardless of the
  identity of the target cell\cite{JCE:86}. In model
  \eqref{eq_dstmodel}, this means each column of $W_D$ is either
  non-negative or non-positive.  To characterize such constraint, a
  feasible formulation is to let $W_D={W}_{V}{W}_{S}$, where
  ${W}_{V}\in\mathbb{R}^{n\times n}$ has non-negative entries, i.e.,
  ${W}_{V}[i,j]\ge0$ for all $i,j\in\until{n}$; and
  ${W}_{S}=\diag\{w_S(1),\dots,w_S(n)\}$ is a diagonal matrix with
  $w_S(i)=\pm 1$. Here, the matrix ${W}_{V}$ encodes the magnitudes of
  entries in $W_D$ whereas the matrix ${W}_{S}$ encodes their signs.
  Since ${W}_{S}$ is combinatorial, a possible way to solve the
  problem is by exhausting all possible combinations. However, the
  complexity of this approach grows exponentially with the number of
  nodes $n$.
  Nevertheless, in computational neuroscience, techniques exist to
  determine the excitatory or inhibitory nature of the neurons by
  classifying their spike wave-forms\footnote{Excitatory neurons have
    slower and wider spikes while inhibitory neurons have faster and
    narrower ones~\cite{EN-JC:21-tacII}.}\cite{RMB-DJS:02}.  Based on
  this, we can assume that ${W}_{S}$ is known a priori.  Then, if we
  still parameterize $W_D$ in the form of \eqref{eq_defzHP}, the Dale's
  law can be represented by an inequality $\widehat{W}_S\bm{h}\le0$, 
  where $\widehat{W}_S\in\mathbb{R}^{n(n-1) \times n(n+m-1)}$ and 
  $n(n-1)$ equals to the number of entries in $W_D$ whose signs are subject to
  constraints. (The diagonal entries of $W_D$ are $0$ and have no 
  constraints). To incorporate this new inequality constraint in  
  formulation \eqref{eq_obj2b}, we only need to solve the optimization
  problem with a new matrix $S$ defined by
  %
  %
  %
  %
  \begin{align}\label{eq_defS2}
    S=\begin{bmatrix}
      -I_{d(\alpha)} & \\
      & \widehat{W}_S
    \end{bmatrix}, 
  \end{align}
  to account for Dale's law.  \hfill $\oldsquare$
\end{remark}

Problem~\eqref{eq_obj2b} is a reformulation for the parameter
identification of system \eqref{eq_dstmodel}. Its objective function
is a non-smooth function of $\alpha$, but smooth in $\xi$. Given
Lemma~\ref{LM_choiceEC}, one approach to find the global minimizer is
to repeatedly solve the problem for each possible value of
$\mathcal{Q}(\alpha)$. However, since the objective function is
piecewise linear, and the dimension of $\xi$ is large, such approach
can be computationally expensive. This motivates further investigating
the characterization of the optimizer of~\eqref{eq_obj2b}.

\subsection{Scalar optimization for parameter
  identification}\label{Sec_Scalar}
Given a fixed $\alpha$, the optimizers of~\eqref{eq_obj2b} are
characterized by the KKT equations,
\begin{subequations}\label{eq_dalekkt}
  \begin{align}
    \mathcal{Q}(\alpha)^{\top}\left(\mathcal{X}^+-\alpha
    \mathcal{X}+\mathcal{Q}(\alpha){\xi}\right)+S\mu=0 ,
    \\
    S{\xi}\le0 ,
    \\
    \mu\ge 0  ,
    \\
    \mu^{\top}S{\xi}=0 ,
  \end{align}  
\end{subequations}
where $\mu$ is the dual variable corresponding to the inequality
constraint. Since $ \mathcal{J}_0$ is a quadratic function of $\xi$
and the constraints are linear, from Slater's condition, strong
duality holds~\cite{SB-LV:09}. Thus, any $\widehat{\xi}$,
$\widehat{\mu}$ satisfying \eqref{eq_dalekkt} gives the minimizer
of~\eqref{eq_obj2b} for the given~$\alpha$.  Now, assuming the matrix
$\left(\mathcal{Q}(\alpha)^{\top}\mathcal{Q}(\alpha)\right)^{-1}$ is
non-singular, the first equation in \eqref{eq_dalekkt} yields
\begin{align}\label{eq_daleoptxi}
  \widehat{\xi}=
  -(\mathcal{Q}(\alpha)^{\top}\mathcal{Q}(\alpha))^{-1}
  (\mathcal{Q}(\alpha)^{\top}(\mathcal{X}^+-\alpha
  \mathcal{X})+S^{\top}\widehat{\mu}) .
\end{align}	
Substituting this into \eqref{eq_obj2b}, one has
\begin{align*}
  \mathcal{J}_0(\alpha)
  &=\frac{\|M(\alpha) \!\left(\mathcal{X}^+\!-\!\alpha 
    \mathcal{X}\right) -
    \mathcal{Q}(\alpha)
    \!\left(\mathcal{Q}(\alpha)^{\top}\!\mathcal{Q}(\alpha)\right)^{-1}
    \!S^{\top}\widehat{\mu}   
    \|_2^2}{2} ,
\end{align*}	
where
$ M(\alpha) = I - \mathcal{Q}(\alpha) \left(\mathcal{Q}(\alpha)^{\top}
  \mathcal{Q}(\alpha)\right)^{-1}\mathcal{Q}(\alpha)^{\top}$.  Note
that $M(\alpha)$ is symmetric and $M(\alpha)\mathcal{Q}(\alpha)=0$.
%
%
Thus, $\mathcal{J}_0(\alpha)$ can also be written as
\begin{align*}
  \mathcal{J}_0(\alpha)\!=\!\frac{\|M(\alpha)\!\left(\mathcal{X}^+\!\!-\!\alpha
  \mathcal{X}\right)\!\|_2^2 \!+\!
  \|\mathcal{Q}(\alpha)\!\left(\mathcal{Q}(\alpha)^{\top}\mathcal{Q}(\alpha)\right)^{-1}\!\!S^{\top}\!\widehat{\mu}
  \|_2^2}{2} .
\end{align*}


Consider now the scalar-variable optimization problem,
\begin{align}\label{eq_reducedJ}
  \min_{\alpha}
  \quad\mathcal{J}(\alpha)
  &=\frac{\|M(\alpha)\left(\mathcal{X}^+-\alpha
    \mathcal{X}\right)\|_2^2}{2}.
\end{align}
Clearly, for any $\alpha$, one has
$\mathcal{J}(\alpha)\le\mathcal{J}_0(\alpha)$, and the equality holds
if $\widehat{\mu}=0$. Now, consider the following two statements:
\begin{enumerate}
\item $\mathcal{J}(\alpha)$ has a unique global minimizer
  $\widehat{\alpha}$;
\item for $\alpha=\widehat{\alpha}$, $\widehat{\xi}$ given by
  \eqref{eq_daleoptxi} and $\widehat{\mu}=0$ solve~\eqref{eq_dalekkt}.
\end{enumerate}
If both statements are true, then
the global minimizer of $\mathcal{J}(\alpha)$ must also be the global
minimizer of $\mathcal{J}_0(\alpha)$. Furthermore, by the KKT
condition and strong duality, $\widehat{\alpha}$ must be the solution
to problem \eqref{eq_obj2b}. By comparing \eqref{eq_obj2b} and
\eqref{eq_reducedJ}, the advantage of the latter is that the
optimization problem is unconstrained, and the dimension of its
variables is reduced from $(1+n(n+m)+d(\alpha))$ to $1$. This kind of
elimination of variables is referred to as separable nonlinear least
squares problems~\cite{AR-PW:80}.

Nevertheless, for the above derivation to hold, we need to address
several challenges. First, our reasoning in \eqref{eq_daleoptxi}
requires $\mathcal{Q}(\alpha)^{\top}\mathcal{Q}(\alpha)$ to be
non-singular, which means that $\mathcal{Q}(\alpha)$ must have full
column rank.
Second, we have assumed in i) that the minimizer
of~\eqref{eq_reducedJ} is unique.  Third, we have assumed in ii) that
$\widehat{\xi}$ given by \eqref{eq_daleoptxi} and $\widehat{\mu}=0$
solve~\eqref{eq_dalekkt}.  
%
%
We tackle each of these
challenges next.

\section{Identification of the Firing Rate Model}\label{SEC_MR}
In this section, we address the challenges outlined in
Section~\ref{Sec_Scalar} regarding the reformulation of the parameter
identification as the scalar optimization
problem~\eqref{eq_reducedJ}. This sets the basis for the design of
the algorithm to identify the parameters of system \eqref{eq_dstmodel}.

\subsection{Establishing the validity of scalar
  optimization}\label{sec:validity-scalar-opt}

We first show that the scalar optimization problem~\eqref{eq_reducedJ}
provides a valid reformulation of the parameter identification
problem.  We make the following assumption.

\begin{assumption}\label{AS_Unique}
  Let ${\alpha}^{\star}$ be the true parameter of system
  \eqref{eq_dstmodel}.  Given the measured system states $\bm{x}_d(k)$
  and system inputs $\bm{u}_d(k)$, $k \in \until{T_d}$, the matrix
  $(I-E({\alpha}^{\star})^2) (I-E({\alpha})^2) \begin{bmatrix}
    \mathcal{X}&\mathcal{P}\end{bmatrix}$ has full column rank for all
  $ \alpha \in (0,1)$\footnote{Since the true ${\alpha}^{\star}$ is
    unknown, the condition is required to hold for all
    $ \alpha \in (0,1)$.}.
\end{assumption}

\begin{remark}\longthmtitle{Validity of
    Assumption~\ref{AS_Unique}}\label{RM_Unique}
  Note that in Assumption \ref{AS_Unique}, the matrices
  $E({\alpha}^{\star})$, $E({\alpha})$, $\mathcal{X}$ and
  $\mathcal{P}$ are associated with the measurement
  data. Specifically, $\mathcal{X}$ and $\mathcal{P}$ are defined from
  $\bm{u}_d(k)$ and $\bm{x}_d(k)$ in \eqref{eq_defXXP};
  $E({\alpha}^{\star})$ and $E({\alpha})$ are implicitly determined by
  $\bm{x}_d(k)$ and $\bm{x}_d^+(k)$ in \eqref{eq_defXXP} and
  \eqref{eq_defcE}. Besides, the row dimension of
  $\left(I-E({\alpha}^{\star})^2\right)
  \left(I-E({\alpha})^2\right) \begin{bmatrix}
    \mathcal{X}&\mathcal{P}\end{bmatrix}$ grows with the number of
  data samples $T_d$.  A sufficient way of checking whether
  Assumption~\ref{AS_Unique} holds without knowing ${\alpha}^{\star}$
  is to compute the column rank of
  $\left(I-E_1)^2\right)\left(I-E_2)^2\right)\begin{bmatrix}\mathcal{X}&\mathcal{P}\end{bmatrix}$
  for all $ E_1, E_2 \in \mathtt{E}$, where
  $\mathtt{E}=\{E(\alpha)~|~ \alpha \in (0,1)\}$ is the set of all
  possible $E(\alpha)$, which is finite.
  The cardinality of $\mathtt{E}$ is therefore bounded by
  $|\mathtt{E}|\le3^{nT_d}$, which grows exponentially with the
  dimension and the number of data sets.  As we show later in the
  proof of Theorem~\ref{TH_alg1}\textbf{b}, an improved bound can be
  obtained as $|\mathtt{E}|\le4{nT_d}+2$, which greatly reduces the
  complexity of validating Assumption~\ref{AS_Unique}. Alternatively, 
  in Section~\ref{sec:probabilistic}, we provide a probabilistic
  criterion to validate Assumption~\ref{AS_Unique}.
  \hfill$\oldsquare$
\end{remark}

The following result establishes that the scalar
optimization~\eqref{eq_reducedJ} is a viable way of finding the
parameters of the system~\eqref{eq_dstmodel}.

\begin{proposition}\longthmtitle{Validity of scalar
    optimization}\label{prop:unique} 
  Under Assumption \ref{AS_Unique}, the following statements hold:
  \begin{enumerate}[label=\textbf{\alph*}.]
  \item $\textit{[Invertibility]}:$ For all $\alpha \in (0,1)$,
    $\mathcal{Q}(\alpha)=\begin{bmatrix} C(\alpha) &
      -\mathcal{P} \end{bmatrix}$ has full column rank;
  \item $\textit{[Uniqueness of minimizer]}:$ The objective function
    $\mathcal{J}(\alpha)$ in \eqref{eq_reducedJ} has a unique
    minimizer $\widehat{\alpha}={\alpha}^{\star}$;
  \item $\textit{[Feasibility and Validity]}:$ For
    $\alpha=\widehat{\alpha}$,
    $\widehat{\xi} =[\widehat{v}^{\top} \; \widehat{\bm{h}}^{\top}
    ]^{\top}$ given by \eqref{eq_daleoptxi} and $\widehat{\mu}=0$
    solve~\eqref{eq_dalekkt}.  Furthermore,
    $\widehat{\bm{h}} = {\bm{h}}^{\star}$, where ${\bm{h}}^{\star}$
    corresponds to the true parameters of system~\eqref{eq_dstmodel}.
  \end{enumerate}
  %
  %
\end{proposition}
\begin{proof}
  \textbf{a}.  From its definition, $C(\alpha)$ must have full column
  rank. Furthermore, since each of its column has exactly one $1$ or
  $-1$, by elementary row operations, the matrix $\mathcal{Q}(\alpha)$
  can be transformed into:
  \begin{align*}
    \widetilde{\mathcal{Q}}(\alpha)=\left[\begin{array}{c|c}
        I&-\widetilde{\mathcal{P}}_I \\\bm{0} &
        -\widetilde{\mathcal{P}}_R
      \end{array}\right] ,
  \end{align*}
  where $-\widetilde{\mathcal{P}}_R$ is composed of some of the rows
  in matrix $-\mathcal{P}$, which are associated with the zero rows of
  $C(\alpha)$ in $\mathcal{Q}(\alpha)$, and $\bm{0}$ is a zero matrix
  of proper size. Clearly, the rank of
  $\widetilde{\mathcal{Q}}(\alpha)$ equals to that of
  $\mathcal{Q}(\alpha)$. Thus, to show $\mathcal{Q}(\alpha)$ has full
  column rank, we only need to show $\widetilde{\mathcal{P}}_R$ has
  full column rank.  By Assumption~\ref{AS_Unique},
  $\left(I-E({\alpha}^{\star})^2\right)
  \left(I-E({\alpha})^2\right) \begin{bmatrix}\mathcal{X}&\mathcal{P}\end{bmatrix}$
  has full column rank for all $ \alpha \in (0,1)$. 
  %
  %
  %
  %
  As a necessary condition, the
  matrix $\left(I-E({\alpha})^2\right)\mathcal{P}$ must also have full
  column rank for all $ \alpha \in (0,1)$. Now, from the definition of
  $E({\alpha})$ in \eqref{eq_defcE}, we know that $E({\alpha})^2$ is a
  diagonal matrix, with either $0$ or $1$ entries. 
  %
  %
  If we left multiplying $\mathcal{P}$ by
  $\left(I-E({\alpha})^2\right)$,
  the rows of $\mathcal{P}$ associated with the $1$ diagonal entries
  of $E({\alpha})^2$ become zero rows; and the remaining rows of
  $\mathcal{P}$ are kept unchanged.  Comparing
  $\left(I-E({\alpha})^2\right)\mathcal{P}$ and
  $\widetilde{\mathcal{P}}_R$, we observe that the two matrices share
  exactly the same non-zero rows.  Since
  $\left(I-E({\alpha})^2\right)\mathcal{P}$ has full column rank, it
  follows that $\widetilde{\mathcal{P}}_R$ must also have full column
  rank. This shows that $\widetilde{\mathcal{Q}}(\alpha)$, and hence
  $\mathcal{Q}(\alpha)$, has full column rank for
  all~$\alpha \in (0,1)$.

  \smallskip
  \noindent \textbf{b}.  We first show that the true parameter
  ${\alpha}^{\star}$ of system \eqref{eq_dstmodel} is a minimizer of
  \eqref{eq_reducedJ}. From equation \eqref{eq_eqcv}, there exists
  ${v}^{\star}$ such that
  $ \mathcal{X}^+-{\alpha}^{\star} \mathcal{X} +
  C({\alpha}^{\star}){v}^{\star}-\mathcal{P}{\bm{h}^{\star}}=0$.
  From~\eqref{eq_defQxi}, this is equivalent to
  \begin{align}\label{eq_XQxi}
    \mathcal{X}^+-{\alpha}^{\star}
    \mathcal{X}= - Q({\alpha}^{\star})\xi^{\star} .
  \end{align}
  Using the fact that $M({\alpha})Q({\alpha})=0$, we obtain
  \begin{align}\label{eq_Japlhavhstar}
    M(\alpha^{\star})\left(\mathcal{X}^+-\alpha^{\star}  
      \mathcal{X}\right)=-M(\alpha^{\star})Q({\alpha}^{\star})\xi^{\star}=0.
  \end{align}
  From~\eqref{eq_reducedJ}, one has $\mathcal{J}({\alpha}^{\star})=0$.
  Since $\mathcal{J}(\alpha)\ge0$, $\widehat{\alpha}={\alpha}^{\star}$
  is a minimizer of~\eqref{eq_reducedJ}.

  Next, we show that the minimizer of~\eqref{eq_reducedJ} is
  unique. Let $\widehat{\alpha}$ be any minimizer
  of~\eqref{eq_reducedJ}.  Then
  \begin{align}\label{eq_tideaplhavh}
    M(\widehat\alpha)\left(\mathcal{X}^+-\widehat\alpha 
    \mathcal{X}\right)=0.
  \end{align} 
  Subtracting equations \eqref{eq_Japlhavhstar} and
  \eqref{eq_tideaplhavh}, and using the definition of
  $M({\alpha})=I - \mathcal{Q}(\alpha)
  \left(\mathcal{Q}(\alpha)^{\top}
    \mathcal{Q}(\alpha)\right)^{-1}\mathcal{Q}(\alpha)^{\top}$, yields
  \begin{align*}
    \left(\widehat{\alpha}\!-\!{\alpha}^{\star}\right)
    \mathcal{X} +
    \mathcal{Q}(\widehat\alpha)\theta(\widehat\alpha) -
    \mathcal{Q}(\alpha^{\star})\theta(\alpha^{\star}) 
    =0 ,
  \end{align*}
  where
  $\theta(\alpha)=\left(\mathcal{Q}(\alpha)^{\top}
    \mathcal{Q}(\alpha)\right)^{-1}\mathcal{Q}(\alpha)^{\top}\left(\mathcal{X}^+
    - \alpha \mathcal{X}\right)$.  We multiply this equation on the
  left
  by the matrices $\left(I-E({\alpha}^{\star})^2\right)$ and
  $\left(I-E(\widehat{\alpha})^2\right)$, which are diagonal and hence
  commute, to obtain
  \begin{align}\label{eq_hataplhavh2}
    &\left(I-E({\alpha}^{\star})^2\right)\left(I-E(\widehat{\alpha})^2\right)
      \mathcal{X}\left(\widehat{\alpha}-{\alpha}^{\star}\right)\nonumber
    \\
    &\quad + \left(I-E({\alpha}^{\star})^2\right)
      \left(I-E(\widehat{\alpha})^2\right) \mathcal{Q}(\widehat\alpha)\theta(\widehat\alpha)
      \nonumber
    \\
    &\quad-\left(I-E(\widehat{\alpha})^2\right)
      \left(I-E({\alpha}^{\star})^2\right)
      \mathcal{Q}(\alpha^{\star})\theta(\alpha^{\star})=0.  
  \end{align}
  Using~\eqref{eq_CTC}, we have
  $\left(I-E({\alpha})^2\right)\!C({\alpha})\! =
  \left(I-C({\alpha})C({\alpha})^{\top}\right)\!C({\alpha})\! =
  \!C({\alpha})-C({\alpha})I = 0$. Thus,
  \begin{align}\label{eq_EQT}
    \left(I-E({\alpha})^2\right)\mathcal{Q}(\alpha)\theta(\alpha)
    &=\left(I-E({\alpha})^2\right)\begin{bmatrix} C(\alpha) & -\mathcal{P}
    \end{bmatrix}\theta(\alpha)\nonumber\\
    &=\left(I-E({\alpha})^2\right)\begin{bmatrix} \bm{0} & -\mathcal{P}
    \end{bmatrix}\theta(\alpha)\nonumber\\
    &=-\left(I-E({\alpha})^2\right)\mathcal{P}~
      \widetilde\theta(\alpha) ,
  \end{align}
  where $\widetilde\theta(\alpha)=\begin{bmatrix}
    \bm{0} &  I_{n(n+m-1)}
  \end{bmatrix} \theta(\alpha)$.  Using~\eqref{eq_EQT}
  in~\eqref{eq_hataplhavh2},
  \begin{align*}
    \left(I-E({\alpha}^{\star})^2\right)\left(I-E(\widehat{\alpha})^2\right)
    \begin{bmatrix} 
      \mathcal{X}& \mathcal{P}
    \end{bmatrix}
    \begin{bmatrix}[1.3]
      \widehat{\alpha}-\alpha^{\star}\\
      \widetilde\theta(\alpha^{\star})-\widetilde\theta(\widehat\alpha)
    \end{bmatrix}=0.
  \end{align*}
  Using Assumption~\ref{AS_Unique},
  we deduce that $\widehat{\alpha}-{\alpha}^{\star}=0$ and thus the
  minimizer of~\eqref{eq_reducedJ} is unique.

  \noindent \textbf{c}.  Let
  $\alpha = \widehat{\alpha} = \alpha^\star$.  Since
  $\mathcal{Q}(\alpha^{\star})^{\top}\mathcal{Q}(\alpha^{\star})$ is
  non-singular, \eqref{eq_daleoptxi} is equivalent to the first
  equation in \eqref{eq_dalekkt}. If $\widehat{\mu}=0$, the last two
  equations in \eqref{eq_dalekkt} are automatically satisfied and
  $\widehat{\xi}$ takes the form
  $ \widehat{\xi}= -\left(\mathcal{Q}(\alpha^{\star})^{\top}
    \mathcal{Q}(\alpha^{\star})\right)^{-1}\mathcal{Q}(\alpha^{\star})^{\top}
  \left(\mathcal{X}^+-\alpha^{\star} \mathcal{X}\right)$.  Left
  multiplying by $Q(\alpha^{\star})^{\top}$ on \eqref{eq_XQxi}, we
  obtain
  \begin{align*}
    {\xi^{\star}}= -\left(\mathcal{Q}(\alpha^{\star})^{\top}
    \mathcal{Q}(\alpha^{\star})\right)^{-1}
    \mathcal{Q}(\alpha^{\star})^{\top}\left(\mathcal{X}^+ -
    \alpha^{\star} \mathcal{X}\right) = \widehat{\xi}.
  \end{align*}
  Therefore, $S\widehat{\xi} \le 0$, and hence $\widehat{\xi}$ and
  $\hat{\mu}=0$ solve~\eqref{eq_dalekkt}.  Finally,
  $\widehat{\xi} = \xi^{\star}$ implies
  $\widehat{\bm{h}} = {\bm{h}}^{\star}$, completing the proof.
\end{proof}

\subsection{Probabilistic condition for validating
  Assumption~\ref{AS_Unique}}\label{sec:probabilistic}

Assumption~\ref{AS_Unique} is critical for establishing
Proposition~\ref{prop:unique}. However, according to
Remark~\ref{RM_Unique}, verifying it directly can be computationally
expensive.  Here, we discuss a probabilistic condition under which
Assumption~\ref{AS_Unique} holds.

From model~\eqref{eq_dstmodel}, since $\alpha\in(0,1)$, we know that
the system states must be bounded by
$\bm{x}_{\max}={s_D}/{(1-\alpha)}$.  Based on this, we make the
following assumption.

\begin{assumption}\label{AS_excitation}
  The data samples $\bm{x}_d(k)$, $\bm{x}^+_d(k)$, $\bm{u}_d(k)$,
  $k \in \until{T_d}$, have the following statistical properties:
  \begin{enumerate}[label=\textbf{\alph*}.]
  \item $\textit{[Linear independence]}:$ 
    For any $\Omega\subset \until {T_d}$ with $|\Omega|=m+n$, the set
    of vectors $\{\bm{p}_d(k)~|~k\in\Omega\}$ is linearly independent,
    where $\bm{p}_d(k)=\col\{\bm{x}_d(k),\bm{u}_d(k)\}$;
    %
    %

  \item $\textit{[Distribution of variables]}:$ The collected data samples 
  $\bm{x}_d(k)\in[0,~\bm{x}_{\max}]^n$ are independent and
    identically distributed (i.i.d.) in terms of $k\in\until {T_d}$; and
    $\bm{u}_d(k)\in[\bm{u}_{\min},~\bm{u}_{\max}]^m$ are i.i.d. in terms 
    of $k\in\until {T_d}$; 
    %
    %
		
  \item $\textit{[Probability]}:$ There exists $\gamma>0$ such that,
    $\forall k\in\until {T_d}$, the distribution of data samples
    satisfy:
    \begin{align}\label{eq_probxd}
      &\Pr\left(\maxv(\bm{x}^+_d(k)-\bm{x}_d(k))\ge
        \gamma\right)=\sigma_1>0 , \nonumber
      \\
      &\Pr\left(\minv(\bm{x}^+_d(k)\!-\!\bm{x}_d(k))>0 \
        \textstyle{\bigwedge}\
        \maxv(\bm{x}^+_d(k))<\gamma\right)\nonumber
      \\ 
      &~~~~~~~~=\sigma_2>0 .
    \end{align}
  \end{enumerate}
\end{assumption}
%
%


\begin{remark}\longthmtitle{Validity of Assumption~\ref{AS_excitation}}
  One can interpret statement \textbf{a} as a variation of the
  \textit{persistent excitation} condition~\cite{JCW-PR-IM-BLMDM:05},
  which is a widely used assumption in system identification and
  data-driven control. This statement is generically true in the sense
  that the set of vectors which do not satisfy
  Assumption~\ref{AS_excitation}\textbf{a} has zero Lebesgue measure
  in $\mathbb{R}^{m+n}$.  Statement~\textbf{b} is satisfied if the
  sampling times are taken randomly from multiple system trajectories
  under dynamics~\eqref{eq_dstmodel}. For statement \textbf{c},
  since $\bm{x}^+_d(k)$ is a deterministic function of two i.i.d. variables 
  $\bm{x}_d(k)$ and $\bm{u}_d(k)$, then $\bm{x}^+_d(k)$ are also
  i.i.d. in terms of $k\in\until {T_d}$. 
  Given dynamics~\eqref{eq_dstmodel}, as long as the distributions of 
  the variables are strictly positive (i.e., all states are possible), 
  there must exist $\gamma$ such that the two probabilities in 
  \eqref{eq_probxd} are strictly positive.  \hfill$\oldsquare$
\end{remark}

Based on Assumption~\ref{AS_excitation}, the following result 
describes a probability for Assumption~\ref{AS_Unique} to hold. 

\begin{proposition}\label{PR_probU}
  Given Assumption~\ref{AS_excitation}, let $\rho$ denote the
  probability that Assumption~\ref{AS_Unique} holds. Then,
  \begin{align*}
    \rho\ge\left(\!1-(1-\sigma_1)^{\left\lfloor
    \frac{T_d}{2}\right\rfloor}\right) \!\left[
    \sum_{\ell=m+n}^{\left\lceil \frac{T_d}{2}\right\rceil}
    {\left\lceil \frac{T_d}{2}\right\rceil \choose
    \ell}\sigma_2^{\ell}(1-\sigma_2)^{\left\lceil
    \frac{T_d}{2}\right\rceil-\ell} \right]
  \end{align*}
  Furthermore, as $T_d\to\infty$, $\rho$ converges to $1$
  exponentially fast.
\end{proposition}
	

\begin{proof}
  Assumption~\ref{AS_Unique} requires that the matrix
  $(I-E({\alpha}^{\star})^2) (I-E({\alpha})^2) \begin{bmatrix}
  	\mathcal{X}&\mathcal{P}\end{bmatrix}$ has full column rank for all
  $ \alpha \in (0,1)$. To verify this,
  we start by introducing a partition of the matrix
  $\begin{bmatrix}\mathcal{X}&\mathcal{P}\end{bmatrix}$. Recall from
  definition \eqref{eq_defXXP} that $\mathcal{P}$ is a column stack of
  $\bm{P}_d(k)$ for $k\in\until{T_d}$. Since each $\bm{P}_d(k)$ is
  block diagonal, its columns are inherently linearly independent. Based on
  this observation, we partition the matrix $\mathcal{P}$ into $n$
  column blocks of size $\mathbb{R}^{nT_d\times (n+m-1)}$ as follows
  \begin{align*}
    \mathcal{P}
    &=
      \text{\scriptsize $\left[\begin{array}{c |c |c |c}
                            \bar{\bm{p}}^{\top}_1(1)& & &\\
                                                    & \bar{\bm{p}}^{\top}_2(1)& &\\
                                                    & & \ddots&\\
                                                    & & &\bar{\bm{p}}^{\top}_n(1)\\
                            \bar{\bm{p}}^{\top}_1(2)& & &\\
                                                    & \bar{\bm{p}}^{\top}_2(2)& &\\
                                                    & & \ddots&\\
                                                    & & &\bar{\bm{p}}^{\top}_n(2)\\
                            \bm{\vdots}& \bm{\vdots}& &\bm{\vdots}\\
                            \bar{\bm{p}}^{\top}_1({T_d})& & &\\
                                                    & \bar{\bm{p}}^{\top}_2({T_d})& &\\
                                                    & & \ddots&\\
                                                    & & &\bar{\bm{p}}^{\top}_n({T_d})\\
        \end{array}
      \right]$ }
  \end{align*}
  where we can conveniently write as
  $\mathcal{P} = [ \mathcal{P}_1 \; \mathcal{P}_2 \; \cdots \;
  \mathcal{P}_n] $, and study the column rank of each block
  independently.  This also induces a decomposition of $\mathcal{X}$ corresponding to the non-zero rows of $\mathcal{P}_i$, which leads to vectors $\mathcal{X}_i$, $i\in\until{n}$ as
  \begin{align*}
    \mathcal{X}_1 &= \text{\scriptsize
                    $\left[\begin{matrix}[1.2]
                        \mathcal{X}[1]\\
                        0\\
                        \vdots\\
                        0\\
                        \mathcal{X}[n+1]\\
                        0\\
                        \vdots\\
                        0\\
                        \bm{\vdots}\\
                        \mathcal{X}[(T_d-1)n+1]\\
                        0\\
                        \vdots\\
                        0\\
		\end{matrix}
		\right]$	},
	\; \mathcal{X}_2 = \text{\scriptsize $\left[\begin{matrix}[1.2]
			0\\
			\mathcal{X}[2]\\
			\vdots\\
			0\\
			0\\
			\mathcal{X}[n+2]\\
			\vdots\\
			0\\
			\bm{\vdots}\\
			0\\
			\mathcal{X}[(T_d-1)n+2]\\
			\vdots\\
			0\\
		\end{matrix}
		\right]$	},\bm{\cdots}    
  \end{align*}
  Now, due to the
  row separation, if for all $i\in\until{n}$, the matrices
  $\begin{bmatrix}\mathcal{X}_i&\mathcal{P}_i\end{bmatrix}\in\mathbb{R}^{nT_d\times
    (n+m)}$ have full column rank, then the matrix
  $\begin{bmatrix}\mathcal{X}&\mathcal{P}\end{bmatrix}$ also has full
  column rank. Furthermore, due to the diagonal structure of
  $\left(I-E({\alpha})^2\right)$, if all matrices
  $\left(I-E({\alpha}^{\star})^2\right)
  \left(I-E({\alpha})^2\right)\begin{bmatrix}\mathcal{X}_i&
    \mathcal{P}_i\end{bmatrix}\in\mathbb{R}^{nT_d\times (n+m)}$,
  $i\in\until{n}$ have full column rank, then the matrix
  $\left( I-E({\alpha}^{\star})^2\right)
  \left(I-E({\alpha})^2\right)\begin{bmatrix}\mathcal{X}&\mathcal{P}\end{bmatrix}$
  has full column rank.

  To proceed, consider the rank of matrix
  $\left(I-E({\alpha}^{\star})^2\right)$ $\left(I-E({\alpha})^2\right)\begin{bmatrix}\mathcal{X}_i&\mathcal{P}_i\end{bmatrix}$.
  By definition,
  $\begin{bmatrix}\mathcal{X}_i&\mathcal{P}_i\end{bmatrix}$ is sparse,
  with only $T_d$ number of non-zero rows. Furthermore, each non-zero row is
  associated with one piece of data for $k\in\until{T_d}$, and has the
  following form
  \begin{align}\label{eq_rowXP}
    \begin{bmatrix}\mathcal{X}[n(k-1)+i]&\bar{\bm{p}}_i(k)\end{bmatrix}.
  \end{align}
  Recall that $\bar{\bm{p}}_i(k)=\left(\bm{p}_d(k)\right)_{-i}$, and
  $\mathcal{X}[n(k-1)+i]=\bm{p}_d(k)[i]$ is exactly the entry removed 
  from $\bm{p}_d(k)$. Therefore, the non-zero rows of
  $\begin{bmatrix}\mathcal{X}_i&\mathcal{P}_i\end{bmatrix}$ can be
  obtained by performing an elementary column operation on the matrix
  $\col\{\bm{p}^{\top}_d(1),\cdots,\bm{p}^{\top}_d({T_d})\}$ which 
  shifts its $i$th column to the first position.  This fact,
  together with Assumption \ref{AS_excitation}\textbf{a}, indicates
  that any $n+m$ number of non-zero rows in
  $\begin{bmatrix}\mathcal{X}_i&\mathcal{P}_i\end{bmatrix}$ are
  linearly independent.  Now, from the definition of $E({\alpha})$ in
  \eqref{eq_defcE}, we know that
  $\left(I-E({\alpha})^2\right)\left(I-E({\alpha}^{\star})^2\right)$
  is a diagonal matrix with $0$, $1$ entries.  To make sure
  $\left(I-E({\alpha}^{\star})^2\right)\left(I-E({\alpha})^2\right) \begin{bmatrix}
    \mathcal{X}_i & \mathcal{P}_i\end{bmatrix}$ has full column rank,
  a sufficient condition is that at least $m+n$ number of non-zero
  rows \eqref{eq_rowXP} in
  $\begin{bmatrix}\mathcal{X}_i&\mathcal{P}_i\end{bmatrix}$ are
  associated with the entry~$1$ in
  $\left(I-E({\alpha}^{\star})^2\right)\left(I-E({\alpha})^2\right)$. This calls 
  for the following derivation.
	
 Based on the definition of $E(\alpha)$ in \eqref{eq_defcE}, the diagonal 
 entries in $\left(I-E({\alpha}^{\star})^2\right)\left(I-E({\alpha})^2\right)$
 equal to $1$ if and only if the corresponding entries in $E({\alpha}^{\star})$
 and $E({\alpha})$ are zeros. That is, the first two conditions
 in  \eqref{eq_defcE} must not hold. Since these conditions involve $\maxv(\mathcal{X}^+-\alpha
 \mathcal{X})$, we first assume $\maxv(\mathcal{X}^+-\alpha
 \mathcal{X})$ has a lower bound $\gamma$ and computes a probability 
 such that this bound is valid. To this end,
  %
  %
    %
  %
  from Assumption~\ref{AS_excitation}\textbf{c}, one has
  \begin{align}\label{eq_Pr1} \nonumber
    \Pr \big( \! \max_{k \in\left\{ 1,\dots, \left\lfloor
      \frac{T_d}{2}\right\rfloor\right\}}&
      (\maxv(\bm{x}^+_d(k) \! -\!\bm{x}_d(k)))\ge\gamma  \big)\\
      &= 1 -(1-\sigma_1)^{\left\lfloor \frac{T_d}{2}\right\rfloor}, 
  \end{align}
  where the probabilities are multiplicable due to the data
  independence in Assumption \ref{AS_excitation}\textbf{b}. 
  Furthermore, since $\alpha\in(0,1)$, and $\bm{x}_d(k)\ge0$, by definition,
  \begin{align*}
  	\maxv(\mathcal{X}^+-\alpha\mathcal{X})
  	&=\max_{k \in\left\{ 1,\dots, {T_d}\right\}}
  	(\maxv(\bm{x}^+_d(k) -\alpha\bm{x}_d(k)))\\
  	&\ge\max_{k \in\left\{ 1,\dots, \left\lfloor
  		\frac{T_d}{2}\right\rfloor\right\}}
  	(\maxv(\bm{x}^+_d(k) -\bm{x}_d(k))).
  \end{align*}
 This, together with \eqref{eq_Pr1} yields
 \begin{align}\label{eq_Pr1_1}
 	\Pr \big(\maxv(\mathcal{X}^+-\alpha\mathcal{X})\ge\gamma  \big)\ge 1 -(1-\sigma_1)^{\left\lfloor \frac{T_d}{2}\right\rfloor}. 
 \end{align}
 Equation \eqref{eq_Pr1_1} gives the probability that $\maxv(\mathcal{X}^+-\alpha
 \mathcal{X})$ is lower bounded by $\gamma$.
  %
  %
  Now, to verify that the first two conditions in \eqref{eq_defcE} do not hold, 
  we define a set
  %
  %
  \begin{align*}
    \mathcal{B}\!=\!\left\{k \left|\  \begin{matrix}
          \minv(\bm{x}^+_d(k)\!-\!\bm{x}_d(k))>0 \ \textstyle{\bigwedge}\
          \maxv(\bm{x}^+_d(k))<\gamma,
          \\
          \textstyle k \in\left\{\left\lfloor
            \frac{T_d}{2}\right\rfloor+1,\dots, T_d\right\}
        \end{matrix}\right.
    \right\}. 
  \end{align*}
  Clearly, for $\alpha\in(0,1)$ and any $k \in\left\{ 1,\dots, {T_d}\right\}$, there holds
  \begin{align*}
    \minv(\bm{x}^+_d(k)\!-\!\bm{x}_d(k))\le\bm{x}^+_d(k)-\alpha\bm{x}_d(k)
    \le
    \maxv(\bm{x}^+_d(k))
  \end{align*}
  Thus, for any $k\in\mathcal{B}$, based on \eqref{eq_defcE}, the 
  corresponding entries in $E({\alpha}^{\star})$
  and $E({\alpha})$ are zeros.
  Furthermore, from the second equation of Assumption 
  \ref{AS_excitation}\textbf{c}, one has
  \begin{align}\label{eq_Pr2}
    \Pr\left(|\mathcal{B}|\ge m+n\right)=\sum_{\ell=m+n}^{\left\lceil
    \frac{T_d}{2}\right\rceil} {\left\lceil
    \frac{T_d}{2}\right\rceil \choose
    \ell}\sigma_2^{\ell}(1-\sigma_2)^{\left\lceil
    \frac{T_d}{2}\right\rceil-\ell},
  \end{align}
  %
  %
  which describes the probability when at least $m+n$ number of non-zero rows 
  \eqref{eq_rowXP} in $\begin{bmatrix}\mathcal{X}_i&\mathcal{P}_i\end{bmatrix}$ 
  are associated with $1$ entries in
  $\left(I-E({\alpha})^2\right)\left(I-E({\alpha}^{\star})^2\right)$.
  Since the construction of $\mathcal{B}$
  builds on $\gamma$, by combining~\eqref{eq_Pr1_1} and~\eqref{eq_Pr2}, one has
  \begin{align*}
    \rho\ge\left( 1-(1-\sigma_1)^{\left\lfloor
    \frac{T_d}{2}\right\rfloor}\right)\! \left[
    \sum_{\ell=m+n}^{\left\lceil \frac{T_d}{2}\right\rceil} {\left\lceil
    \frac{T_d}{2}\right\rceil \choose
    \ell}\sigma_2^{\ell}(1-\sigma_2)^{\left\lceil
    \frac{T_d}{2}\right\rceil-\ell} \right]
  \end{align*}
	where all the matrices
	$\left(I-E({\alpha})^2\right)
	\left(I-E({\alpha}^{\star})^2\right) \begin{bmatrix}\mathcal{X}_i&\mathcal{P}_i\end{bmatrix}$,
	$i\in\until{n}$ have full column rank. This is the first 
	statement in Proposition \ref{PR_probU}.
  
  To proceed, we show that $\rho\to 1$ exponentially fast as
  $T_d\to\infty$. This is clearly the case for the first term
  $( 1-(1-\sigma_1)^{\left\lfloor \frac{T_d}{2}\right\rfloor})$. 
  Additional, note that
  \begin{align*}
    &\sum_{\ell=m+n}^{\left\lceil \frac{T_d}{2}\right\rceil}
      {\left\lceil \frac{T_d}{2}\right\rceil \choose
      \ell}\sigma_2^{\ell}(1-\sigma_2)^{\left\lceil
      \frac{T_d}{2}\right\rceil-\ell}\nonumber
    \\
    =\ &1 - \sum_{\ell=0}^{m+n-1} {\left\lceil \frac{T_d}{2}\right\rceil
         \choose \ell}\sigma_2^{\ell}(1-\sigma_2)^{\left\lceil
         \frac{T_d}{2}\right\rceil-\ell}\nonumber
    \\
    \ge\ & 1- \sum_{\ell=0}^{m+n-1} \left\lceil
           \frac{T_d}{2}\right\rceil ^{\ell}
           \sigma_2^{\ell}(1-\sigma_2)^{\left\lceil
           \frac{T_d}{2}\right\rceil-\ell}\nonumber
    \\
    \ge\ & 1- (m+n-1) \left\lceil \frac{T_d}{2}\right\rceil ^{m+n-1}
           (1-\sigma_2)^{\left\lceil \frac{T_d}{2}\right\rceil-\ell}  .
  \end{align*}
  Since $m$, $n$ are constants and $1-\sigma_2<1$, this converges to
  $1$ exponentially fast, completing the proof.
\end{proof}

Note that by examining the indexes of $k$ in ~\eqref{eq_Pr1_1}
and~\eqref{eq_Pr2} respectively, we use the first half of data samples
to determine a lower bound $\gamma$ of
$\maxv(\mathcal{X}^+-\alpha\mathcal{X})$, then use the second half of
data samples to guarantee the non-zero entries in
$\left(I-E({\alpha}^{\star})^2\right)\left(I-E({\alpha})^2\right)$.


\subsection{Algorithm for parameter identification}

Given our discussion in Sections~\ref{sec:validity-scalar-opt}
and~\ref{sec:probabilistic}, all the parameters of
system~\eqref{eq_dstmodel} can be determined by solving the
minimization~\eqref{eq_reducedJ}.  The latter is challenging given the
piecewise-constant nature of $M(\alpha)$ as a function of $\alpha$,
which in general makes $\Jc$ discontinuous and non-convex. Here we
tackle this problem and design an algorithm to solve it.

We start by observing that the feasible region of~\eqref{eq_reducedJ}
can be refined. From \eqref{eq_compactdataeq}, we know
$\mathcal{X}^+-\alpha \mathcal{X} =
\left[\mathcal{P}\bm{h}\right]_0^{s_{_D}}\ge0$. Thus, given data sets
$\mathcal{X}^+$ and $\mathcal{X}$, the feasible region of $\alpha$ can
be shrunk to $ (0,\alpha_{\max}]$, where
\begin{align*}
  \alpha_{\max}=
  \min\left(1,\min_{i \in \until{nT_d}, {\mathcal{X}[i]}\neq0}\left(\frac{\mathcal{X}^+[i]}{\mathcal{X}[i]}\right)\right)
\end{align*}
%
%
Note that if $\alpha_{\max}=1$, this procedure actually enlarges the
feasible region of $\alpha$ by adding the point $\alpha=1$. However,
since the extra point has no impact in the result of
Proposition~\ref{prop:unique}\textbf{b}, it does not change the
solution to the optimization problem~\eqref{eq_reducedJ}.

The key idea of the algorithm proposed below to solve the optimization
problem~\eqref{eq_reducedJ} is to identify the domains where
$M(\alpha)$ are constant matrices. Within each domain,
\eqref{eq_reducedJ} is a quadratic optimization problem, so its
solution can be directly obtained.  We then compare all the solutions
to get the global optimum.  In order to do so, the challenge is to
determine the boundary points on $(0,\alpha_{\max}]$ that separate the
domains on which $M(\alpha)$ is constant. As we show next, the
number of boundary points is linear in~$nT_d$.
%
%

\begin{algorithm*}[htb!]
  \setstretch{1.2}
  \label{Algorithm_Main}
  \caption{Scalar Optimization via Domain Partition}
  \SetAlgoLined
  \textbf{Input} $\mathcal{X}^+$, $\mathcal{X}$ and $\mathcal{P}$\;
  {Define} $\mathcal{S}=\mathcal{M}=\mathcal{Z}=\emptyset$;
  $\mathcal{T}= \{1, \dots ,nT_d\}$\;
  {Initial values:} $\psi_0=0$, $\ell=0$\;
  {Initial sets:}
  $\mathcal{S}=\left\{i~\big|~\mathcal{X}^+[i]=\maxv\left(\mathcal{X}^+\right),~
    i\in\mathcal{T}\right\}$;
  $\mathcal{Z}=\left\{i~\big|~\mathcal{X}^+[i]=0,~
    i\in\mathcal{T}\right\}$;
  $\mathcal{M} = \mathcal{T} \setminus
  (\mathcal{S}\bigcup\mathcal{Z})$\;
	
  \While{$\psi_{\ell}<\alpha_{\max}$}
  {Find the smallest
    $\widehat\psi>\psi_{\ell}$ such that
    $\displaystyle\max_{j\in\mathcal{M}} (\mathcal{X}^+[j] -
    \widehat\psi\mathcal{X}[j])=\maxv(\mathcal{X}^+ -
    \widehat\psi\mathcal{X})$ or $\displaystyle\min_{j\in\mathcal{M}}
    (\mathcal{X}^+[j] -
    \widehat\psi\mathcal{X}[j])=0$ \tcp*{\small \color{blue}
      $C(\alpha)$ takes different values for
      $\alpha$ on different sides of
      $\widehat\psi$} \label{step_find} Let
    $\psi_{\ell+1}=\widehat\psi$\;
		
    Compute $\displaystyle
    C_{A\ell}=C\big(\frac{\psi_{\ell}+\psi_{\ell+1}}{2}\big)$, 
    $\mathcal{Q}_{A\ell}= [C_{A\ell} \;
      -\mathcal{P}] $, and
      $M_{A\ell}=I-\mathcal{Q}_{A\ell}\left(\mathcal{Q}_{A\ell}^{\top}\mathcal{Q}_{A\ell}\right)^{-1}\mathcal{Q}_{A\ell}^{\top}$
      \label{step_optstart}
      \tcp*{\small \color{blue} $\mathcal{Q}_{A\ell}$ is
      constant  for $\alpha\in(\psi_{\ell},\psi_{\ell+1})$} 
    
    Solve $\displaystyle  \widehat{\alpha}_{A\ell}=
    \underset{{\alpha_{A\ell}\in(\psi_{\ell},\psi_{\ell+1})}}{\arg\minv}
    ~\frac{1}{2}\|M_{A\ell}\left(\mathcal{X}^+-\alpha_{A\ell}
      \mathcal{X}\right)\|_2^2$  \label{step_optend} \tcp*{\small
      \color{blue} Solve the optimization problem} 
    
    Compute ${\mathcal{J}}(\widehat{\alpha}_{A\ell})=\frac{1}{2}\|M_{A\ell}\left(\mathcal{X}^+-\widehat{\alpha}_{A\ell} \mathcal{X}\right)\|_2^2$\;
    
    Compute $\displaystyle C_{B\ell}=C(\psi_{\ell+1})$,
    $\mathcal{Q}_{B\ell}=\begin{bmatrix}C_{B\ell} &
      -\mathcal{P}\end{bmatrix} $, and
    $M_{B\ell}=I-\mathcal{Q}_{B\ell}\left(\mathcal{Q}_{B\ell}^{\top}\mathcal{Q}_{B\ell}\right)^{-1}\mathcal{Q}_{B\ell}^{\top}$
     \label{step_cpstart}\;
    
    Let $\widehat{\alpha}_{B\ell}=\psi_{\ell+1}$. Compute ${\mathcal{J}}(\widehat{\alpha}_{B\ell})=\frac{1}{2}\|M_{B\ell}\left(\mathcal{X}^+-\widehat{\alpha}_{B\ell} \mathcal{X}\right)\|_2^2$  \label{step_cpend}\;
    
    Update $\mathcal{S}=\left\{i~\big|~ \left(\mathcal{X}^+[i] - \psi_{\ell+1}\mathcal{X}[i]\right)=\maxv\left(\mathcal{X}^+ - \psi_{\ell+1}\mathcal{X}\right),~ i\in\mathcal{T}\right\}$; 
    $\mathcal{Z}=\left\{i~\big|~ \left(\mathcal{X}^+[i] -
        \psi_{\ell+1}\mathcal{X}[i]\right)=0 ,~
      i\in\mathcal{T}\right\}$;
      $\mathcal{M} = \mathcal{T} \setminus
      (\mathcal{S}\bigcup\mathcal{Z})$   \tcp*{\small \color{blue}Update sets for $\alpha=\psi_{\ell+1}$}\label{step_update} 	
    
    $\ell=\ell+1$ \label{step_itend}\;
  }
  \textbf{Output} $\displaystyle \widehat{\alpha}=  \underset{\alpha\in\{\widehat{\alpha}_{A\ell}\}\bigcup\{\widehat{\alpha}_{B\ell}\}}{\arg\minv} ~    \mathcal{J}(\alpha)$
\end{algorithm*}
%

%
%

\begin{theorem}\longthmtitle{Properties of
    Algorithm~\ref{Algorithm_Main}}\label{TH_alg1}
  Suppose Assumption~\ref{AS_Unique}
  holds. Algorithm~\ref{Algorithm_Main} has the following properties:
  \begin{enumerate}[label=\textbf{\alph*}.]
  \item $\textit{[Minimizer]}$ The output value $\widehat{\alpha}$ is
    the minimizer of problem~\eqref{eq_reducedJ};
  \item $\textit{[Complexity]}$ Algorithm \ref{Algorithm_Main}
    terminates in at most $2nT_d+1$ number of iterations.  The
    computational complexity of the algorithm is
    $\mathcal{O}(nT_d)^{3.34}$, where $n$ is the number of 
    nodes and $T_d$ is the number of sampled data;
  \item $\textit{[Identification]}$ Given 
    $\widehat{\alpha}=\alpha^{\star}$, the variables ${v}^{\star}$ and
    ${\bm{h}}^{\star}$ can be computed as
    \begin{align*}
      \begin{bmatrix}[1.3]
        {v}^{\star}\\ {\bm{h}}^{\star}
      \end{bmatrix}
      =-\left(\mathcal{Q}(\alpha^{\star})^{\top}
      \mathcal{Q}(\alpha^{\star})\right)^{-1}
      \mathcal{Q}(\alpha^{\star})^{\top}\left(\mathcal{X}^+-\alpha^{\star}
      \mathcal{X}\right) ,
    \end{align*}
    yielding the matrices $W_D$ and $B_D$ of
    system~\eqref{eq_dstmodel}. Finally, $s_D$ can be estimated by
    $s_D=\maxv\left(\mathcal{X}^+ - \alpha^{\star}\mathcal{X}\right)$.
    If any entry of $\mathcal{P}\bm{h}$ reaches the upper saturation threshold, 
    this estimate is exactly~$s_D$.
  \end{enumerate}
\end{theorem}
\begin{proof}
  \textbf{a}.  Since $M(\alpha)$ is a piecewise constant
  function of $\alpha$, define $\psi_{\ell}\in(0,\alpha_{\max}]$
  as the \textit{critical points}
  %
  %
  such that the matrix $M(\alpha)$ changes its value when $\alpha$
  passes across $\psi_{\ell}$. Then, on each domain
  $(\psi_{\ell},\psi_{\ell+1})$, the matrix $M(\alpha)$ must be a
  constant. In order to determine the values of $\psi_{\ell}$, note
  that $M(\alpha)$ depends on~$E(\alpha)$.  From~\eqref{eq_defcE},
  $E(\alpha)$ is diagonal and its values are determined by the three
  types of entries in
  $\left(\mathcal{X}^+ - \alpha\mathcal{X}\right)$: the ones that
  reach the upper saturation threshold ($E(\alpha)[i,i]=1$ if
  $\left(\!\mathcal{X}^+[i]\!\!-\!\alpha
    \mathcal{X}[i]\right)=\maxv\left(\mathcal{X}^+ -
    \alpha\mathcal{X}\right)$); the ones that reach the zero
  saturation threshold ($E(\alpha)[i,i]=-1$ if
  $\left(\!\mathcal{X}^+[i]\!\!-\!\alpha \mathcal{X}[i]\right)=0$); and
  the ones in between ($E(\alpha)[i,i]=0$).  For convenience of
  presentation, we use $\mathcal{S}$, $\mathcal{Z}$, $\mathcal{M}$ to
  denote the indices of entries corresponding to $1$, $-1$, $0$,
  respectively. Note that $E(\alpha)$ changes as a function of
  $\alpha$ only if the sets $\mathcal{S}$, $\mathcal{M}$,
  $\mathcal{Z}$ change.  To detect such changes, in step
  \ref{step_find} of Algorithm \ref{Algorithm_Main}, we compute the
  smallest value bigger than $\psi_l$
  %
  %
  such that certain entries of the vector
  $\left(\mathcal{X}^+ - \alpha\mathcal{X}\right)$ shift from set
  $\mathcal{M}$ to set $\mathcal{S}$ or set $\mathcal{Z}$
  (correspondingly, some other entries of
  $\left(\mathcal{X}^+ - \alpha\mathcal{X}\right)$ may leave
  %
  %
  the set $\mathcal{S}$ or set $\mathcal{Z}$ and join set
  $\mathcal{M}$). From steps \ref{step_optstart}-\ref{step_optend} of
  Algorithm~\ref{Algorithm_Main}, we take the middle point of
  $(\psi_{\ell},\psi_{\ell+1})$ to find $C(\alpha)$ and then solve the
  optimization problem on this region.  The union of these open sets
  $(\psi_{\ell},\psi_{\ell+1})$ does not include the critical points
  $\psi_{\ell}$: this is because $E(\psi_{\ell})$ is different from
  $E(\alpha)$ evaluated on $\alpha<\psi_{\ell}$ or on
  $\alpha>\psi_{\ell}$ (since two or more entries of
  $\left(\mathcal{X}^+ - \alpha\mathcal{X}\right)$ may simultaneously
  be the greatest/zero entry).  Thus, in steps
  \ref{step_cpstart}-\ref{step_cpend} of
  Algorithm~\ref{Algorithm_Main}, we compute the values of
  $\mathcal{J}({\alpha})$ on such critical points $\psi_{\ell}$
  separately.
  After finding all possible $\psi_{\ell}$ and the corresponding
  minimizers for ${\mathcal{J}}({\alpha})$, we finally compute the
  global minimizer $\widehat{\alpha}$ by comparing all the obtained
  $ \{\mathcal{J}(\alpha)\}_{\alpha\in\{\widehat{\alpha}_{A\ell}\}
    \bigcup \{\widehat{\alpha}_{B\ell}\}}$.

  \smallskip
  \noindent \textbf{b}.  The complexity of
  Algorithm~\ref{Algorithm_Main} depends on the number of critical
  points~$\{\psi_{\ell}\}$. From the discussion in part \textbf{a},
  these points are determined by changes in the composition of the
  sets $\mathcal{S}$, $\mathcal{M}$, and $\mathcal{Z}$. Based on this
  fact, we first consider element exchanges between $\mathcal{S}$ and
  $\mathcal{M}$. Note that
  $\maxv(\mathcal{X}^+-\alpha \mathcal{X})=\max_{i\in\{1,\dots,nT_d\}}(\mathcal{X}^+[i]-\alpha
  \mathcal{X}[i])$.
  %
  %
  Since $\maxv(\mathcal{X}^+-\alpha \mathcal{X})$ is convex on
  $\alpha$, the line $\mathcal{X}^+[i]-\alpha \mathcal{X}[i]$ for a
  given $i$ can intersect it only once, either for a continuous
  interval of $\alpha$ or at a particular point. Thus, for values of
  $\alpha$ on different domains $(\psi_{\ell},\psi_{\ell+1})$, the
  corresponding sets $\mathcal{S}$ do not have the same elements. In
  other words, if $i\in\mathcal{S}$ for
  $\alpha\in(\psi_{\ell_1},\psi_{\ell_1+1})$, then
  $i\notin\mathcal{S}$ for $\alpha\in(\psi_{\ell_2},\psi_{\ell_2+1})$,
  $\ell_1\neq \ell_2$.  Now, since the vector
  $(\mathcal{X}^+-\alpha \mathcal{X})$ has $nT_d$ entries, the element
  exchange between $\mathcal{S}$ and $\mathcal{M}$ can lead to at most
  $nT_d$ number of $\psi_{\ell}$. A similar argument shows that the
  element exchange between $\mathcal{Z}$ and $\mathcal{M}$
  can create at most $nT_d$ number of $\psi_{\ell}$. Bringing these
  two cases together, one has at most $2nT_d$ number of $\psi_{\ell}$
  on $(0,1]$. Furthermore, since the Algorithm~\ref{Algorithm_Main}
  starts with $\psi_{0}=0$, which requires an extra round of
  execution, it terminates in at most $2nT_d+1$ number of
  iterations. Since each iteration of Algorithm~\ref{Algorithm_Main}
  leads to two different $E(\alpha)$ (e.g., the ones obtained for
  $\alpha=\frac{\psi_\ell+\psi_{\ell+1}}{2}$ and
  $\alpha={\psi_{\ell+1}}$), the number of different matrices
  $E(\alpha)$ is bounded by $|\mathtt{E}|\le4{nT_d}+2$.

  Now, to determine the computational complexity of
  Algorithm~\ref{Algorithm_Main}, we evaluate the computational
  complexity of each of its iterations, i.e., step \ref{step_find} to
  step \ref{step_itend}. Particularly, we are interested in steps
  \ref{step_find}, \ref{step_optstart}, and \ref{step_optend}, since
  the other steps are either trivial or simply duplicate one of these
  steps.  For step~\ref{step_find}, one needs to solve
  $\displaystyle\max_{j\in\mathcal{M}} (\mathcal{X}^+[j] -
  \widehat\psi\mathcal{X}[j])=\maxv(\mathcal{X}^+ -
  \widehat\psi\mathcal{X})$ or
  $\displaystyle\min_{j\in\mathcal{M}} (\mathcal{X}^+[j] -
  \widehat\psi\mathcal{X}[j])=0$.  Each process needs to solve at most
  $2(T_d-1)$ linear, single-variable equations, so the complexity is
  $\mathcal{O}(T_d)$.
  For step \ref{step_optstart}, the computational complexity mainly
  comes from matrix inverse and multiplication. Since the row and
  column size of the matrices is no larger than $nT_d$, by using the
  Coppersmith–Winograd algorithm\cite{DC-SW:87}, the complexity is
  characterized by $\mathcal{O}(nT_d)^{2.34}$.  For
  step~\ref{step_optend}, we can write
  $\frac{1}{2}\|M_{A\ell}\left(\mathcal{X}^+-\alpha_{A\ell}
    \mathcal{X}\right)\|_2^2=c_0+c_1\alpha_{A\ell}+c_2
  \alpha_{A\ell}^2$, where
  \begin{align}\label{eq_C012}
    c_0
    &=\frac{1}{2}(M_{A\ell}
      \mathcal{X}^+)^{\top}M_{A\ell}\mathcal{X}^+,
      \;
      c_1=-(M_{A\ell}\mathcal{X}^+)^{\top} 
      (M_{A\ell}\mathcal{X}), \nonumber
    \\
    c_2
    &=\frac{1}{2}(M_{A\ell}\mathcal{X})^{\top}(M_{A\ell}\mathcal{X}).
  \end{align}
  Thus, the optimization problem only requires to find the minimizer
  of a parabola on a given interval, which is $\mathcal{O}(1)$. For
  this step, the major complexity comes from the matrix multiplication
  in \eqref{eq_C012} which, again by the Coppersmith–Winograd
  algorithm, has complexity $\mathcal{O}(nT_d)^{2.34}$.
  Finally, since the number of iterations is bounded by $2nT_d+1$
  times, the computational complexity of
  Algorithm~\ref{Algorithm_Main} is $\mathcal{O}(nT_d)^{3.34})$.

  \smallskip
  \noindent \textbf{c}. The fact that ${\alpha}^{\star}$ and
  ${\bm{h}}^{\star}$ correspond to the parameters of the system
  follows directly from Proposition~\ref{prop:unique}. Regarding
  $s_D$, from~\eqref{eq_compactdataeq}, we see that
  $\maxv\left(\mathcal{X}^+ - \alpha^{\star}\mathcal{X}\right)$ is the
  best estimate of $s_D$ that one can get from the data.  If any entry
  of $\mathcal{P}\bm{h}$ reaches the upper saturation threshold, this
  estimate is exactly~$s_D$.
  %
  %
\end{proof}

From Remark~\ref{RM_Unique} and Proposition \ref{PR_probU}, we see
that more data is is beneficial to make Assumption~\ref{AS_Unique}
hold. On the other hand, according to Theorem~\ref{TH_alg1}\textbf{b},
more data leads to higher computational complexity of
Algorithm~\ref{Algorithm_Main}. In the absence of measurement noise,
it is sufficient to consider the smallest amount of data that
satisfies Assumption~\ref{AS_Unique}.

\section{Impact of Measurement Noise}
The results in Section \ref{SEC_MR} are established based on the
assumption that the measured data $\bm{x}_d(k)$, $\bm{x}^+_d(k)$,
$\bm{u}_d(k)$ do not involve any measurement noise. However in the
context of neuronal activities, the presence of measurement noise is
inevitable. Although data processing methods
\cite{EN-JS-LC-EJC-XH-MAB-ASM-GJP-SDB:20} such as spatial averaging
(average the readings over aggregated neuron groups) or temporal
averaging (average the readings over small time windows) mitigate the
effect of measurement noise, such impact can never be completely
eliminated. Motivated by this, in this section we introduce a modified
algorithm still based on the scalar optimization problem
\eqref{eq_reducedJ} to handle the presence of noise.  We show that if
the sampled data involves bounded noise, the identification error of
the new algorithm is linearly bounded by the magnitude of the
measurement noise.

Suppose we have noisy data $\bm{x}^+_\epsilon(k)$,
$\bm{x}_\epsilon(k)$, $\bm{u}_\epsilon(k)$, such that
\begin{align}
  \bm{x}^+_\epsilon(k)=\bm{x}^+_d(k)
  &+\epsilon_{x^+}(k),\quad
    \bm{x}_\epsilon(k)=\bm{x}_d(k)+\epsilon_x(k) ,
\end{align}
where $\bm{x}^+_d(k), \bm{x}_d(k)$ correspond to the true states of
the system, and $\epsilon_{x^+}(k), \epsilon_x(k)$ are the associated
measurement noises. Suppose
\begin{align*}
  &\bm{u}_\epsilon(k)=\bm{u}_d(k)+\epsilon_u(k) ,
\end{align*}
where $\bm{u}_\epsilon(k)$ is the input data sample and $\bm{u}_d(k)$
is the practical input that was feed into the system. For all noise
terms, we have the following assumption.

\begin{assumption}\label{AS_Err}
  For $k=1,\cdots,T_d$, the infinity norms for all noises are
  bounded by $\overline\epsilon\in\mathbb{R}$, i.e.,
  $\|\epsilon_{x^+}(k)\|_{\infty}\le\overline\epsilon$,
  $\|\epsilon_{x}(k)\|_{\infty}\le\overline\epsilon$, and
  $\|\epsilon_{u}(k)\|_{\infty}\le \overline\epsilon$.
\end{assumption}

\medskip In order to estimate the system parameter under measurement
noise, one can still follow the idea in Section~\ref{SEC_KI} by
formulating the identification as a scalar optimization
problem. 
Recall that in \eqref{eq_defCv}, we introduce the matrix $C(\alpha)$ to
characterize the threshold nonlinearity of the system by detecting
the $\maxv$ or $0$ entries in vector
$\mathcal{X}^+-\alpha \mathcal{X}$. However, in the presence of
measurement noise, the conditions in \eqref{eq_defcE} no longer
strictly hold. Instead, we need to construct a new matrix
$C_{\epsilon}(\alpha)$ with the following relaxation:
\begin{enumerate}
\item Define a diagonal matrix $E_{\epsilon}(\alpha)\in\mathbb{R}^{nT_d\times nT_d}$ such that for all $i=1,2,\cdots, {nT_d}$,
  \begin{align}\label{eq_defcEnoise}
    E_{\epsilon}(\alpha)[i,i]\!\!=\!\!\left\{\begin{matrix*}[l]
        \!\phantom{-}1 ~~ \text{for} \left(\mathcal{X}_\epsilon^+\!\!-\!\alpha \mathcal{X}_\epsilon\right)\![i] \!\ge\! \maxv(\mathcal{X}_\epsilon^+\!\!-\!\alpha \mathcal{X}_\epsilon)\\~~~~~~~~~~~~~~~~~~~~~~~~~~~~~~~~~-2(1+\alpha)\overline\epsilon	\\
        \!-1~~ \text{for}\left(\mathcal{X}_\epsilon^+\!-\!\alpha \mathcal{X}_\epsilon\right)\![i]\!\le\! (1+\alpha)\overline\epsilon\\
        \!\phantom{-}0 ~~  \text{otherwise}
      \end{matrix*}\right.
  \end{align}
\item Obtain $C_{\epsilon}(\alpha)$ from $E_{\epsilon}(\alpha)$ by removing all zero columns in $E_{\epsilon}(\alpha)$.
\end{enumerate}
The relaxed $C_{\epsilon}(\alpha)$ is capable of characterizing the
threshold property in equation \eqref{eq_dataeqf} by taking into
account the impact of
noise. 
For the new $C_\epsilon(\alpha)$ and $E_\epsilon(\alpha)$, we still
assume that Assumption \ref{AS_Unique} holds, and the satisfaction of
this assumption can still be characterized by Proposition
\ref{PR_probU}.  Furthermore, note that if we bring this new
$C_{\epsilon}(\alpha)$ and the noisy data into equation
\eqref{eq_eqcv}, i.e.,
$C_\epsilon({\alpha}){v}=\mathcal{P}_{\epsilon}{\bm{h}}-\mathcal{X}_{\epsilon}^+
+{\alpha} \mathcal{X}_{\epsilon}$, the condition $v\ge0$ may not
strictly hold. Thus, when we solve the optimization problem with
measurement noise, the corresponding constraint can be removed, and
the $S$ in \eqref{eq_defS2} is simplified to
$S_\epsilon=\begin{bmatrix} \bm{0} & \widehat{W}_S
\end{bmatrix}$.

Following from \eqref{eq_obj2b}, a straightforward way to identify the
system parameter with noisy data is to solve the following problem
with modified $\mathcal{Q}_{\epsilon}(\alpha)$ and $S_\epsilon$,
\begin{align}\label{eq_obj2bn}
  \begin{aligned}
    \min&\qquad\mathcal{J}_0(\alpha,\xi)=\frac{1}{2}\|\mathcal{X}_{\epsilon}^+-\alpha \mathcal{X}_{\epsilon}+
    \mathcal{Q}_{\epsilon}(\alpha)\xi
    \|_2^2\\
    \text{s.t.}&\qquad S_{\epsilon}\xi\le0	
  \end{aligned}
\end{align}
Here, we employ an idea similar to the one we used in
Section~\ref{SEC_KI} and introduce a two-step minimization approach to
solve \eqref{eq_obj2bn}. First, consider the scalar optimization
problem
\begin{align}\label{eq_ModifiedJ}
  \min_{\alpha}  \quad
  \mathcal{J}(\alpha)=\frac{1}{2}\|M_{\epsilon}(\alpha)\left(\mathcal{X}_{\epsilon}^+-\alpha
  \mathcal{X}_{\epsilon}\right)\|_2^2, 
\end{align}
where
$M_{\epsilon}(\alpha)=I - \mathcal{Q}_{\epsilon}(\alpha)
\left(\mathcal{Q}_{\epsilon}(\alpha)^{\top}
  \mathcal{Q}_{\epsilon}(\alpha)\right)^{-1}\mathcal{Q}_{\epsilon}(\alpha)^{\top}$,
$\mathcal{Q}_{\epsilon}(\alpha)=\begin{bmatrix} C_{\epsilon}(\alpha) &
  -\mathcal{P}_{\epsilon}
\end{bmatrix}$, and $\mathcal{P}_{\epsilon}$ is defined the same way
as \eqref{eq_defXXP}, but with noisy data. Due to the impact of
measurement noise, when solving \eqref{eq_ModifiedJ}, we modify the
feasible region of $\alpha\in(0,\alpha_{\max}]$ by redefining
$$\alpha_{\max}=
\min\left(1,\min_{i \in \until{nT_d}, {\mathcal{X}_\epsilon[i]-\overline\epsilon}\neq0}\left(\frac{\mathcal{X}_\epsilon^+[i]+\overline\epsilon}{\mathcal{X}_\epsilon[i]-\overline\epsilon}\right)\right)$$
Algorithm
\ref{Algorithm_Main2} presents the pseudocode for the case when the data
has measurement noise.

After obtaining the minimizer $\widehat{\alpha}$ of
\eqref{eq_ModifiedJ} from Algorithm \ref{Algorithm_Main2}, we further
compute
$\widehat{\xi}=\begin{bmatrix} \widehat{v}^{\top} &
  \widehat{\bm{h}}^{\top}
\end{bmatrix}^{\top}$ by
\begin{align}\label{eq_compvh2}
  \widehat{\xi} =
  \underset{S_{\epsilon}\xi\le0}{\arg\min}~\frac{1}{2}\|\mathcal{X}_{\epsilon}^+-\widehat{\alpha}  
  \mathcal{X}_{\epsilon} + 
  \mathcal{Q}_{\epsilon}(\widehat{\alpha})\xi
  \|_2^2			
\end{align}

\begin{remark}\longthmtitle{Two-step minimization}
  In our approach, we first identify $\widehat{\alpha}$ by solving a
  reduced problem \eqref{eq_ModifiedJ}. Second, we bring
  $\widehat{\alpha}$ into \eqref{eq_compvh2} to compute
  $\widehat{\xi}$.  In the presence of measurement noise, it is
  possible that the minimizer of \eqref{eq_ModifiedJ} is different
  from that of \eqref{eq_obj2bn}. However, as we show below, our
  two-step approach provides a valid estimation of the system
  parameters, in the sense that the estimation error is linearly
  bounded by~$\overline\epsilon$.  \hfill$\oldsquare$
\end{remark}


\begin{algorithm*}[htb]
  \setstretch{1.2}
  \label{Algorithm_Main2}
  \caption{Scalar Optimization via Domain Partition under Noise}
  \SetAlgoLined
  \textbf{Input}  $\mathcal{X}_\epsilon^+$, $\mathcal{X}_\epsilon$ and $\mathcal{P}_\epsilon$\;	
  {Define} $\mathcal{S}=\mathcal{M}=\mathcal{Z}=\emptyset$;  $\mathcal{T}= \{1,\cdots,nT_d\}$\;
  {Initial values:} $\psi_0=0$, $\ell=0$\;
  {Initial sets:} $\mathcal{S}=\left\{i~\big|~\mathcal{X}_\epsilon^+[i]\ge\maxv\left(\mathcal{X}_\epsilon^+\right)-2\overline\epsilon,~ i\in\mathcal{T}\right\}$;  $\mathcal{Z}=\left\{i~\big|~\mathcal{X}_\epsilon^+[i]\le\overline\epsilon,~ i\in\mathcal{T}\right\}$; $\mathcal{M}=\left\{i~\big|~i\not\in\mathcal{S}\bigcup\mathcal{Z},~ i\in\mathcal{T}\right\}$ \label{step_init2}\;
  
  \While{$\psi_{\ell}<\alpha_{\max}$}
  { 
    Find the smallest $\widehat\psi>\psi_{\ell}$, such that $\displaystyle \min_{j\in\mathcal{S}} \left(\mathcal{X}_\epsilon^+ [j]- \widehat\psi\mathcal{X}_\epsilon[j]\right)\le \maxv \left(\mathcal{X}_\epsilon^+ - \widehat\psi\mathcal{X}_\epsilon\right)- 2(1+\widehat\psi)\overline\epsilon$ or $\displaystyle \max_{j\in\mathcal{M}} \left(\mathcal{X}_\epsilon^+[j] - \widehat\psi\mathcal{X}_\epsilon[j]\right)\ge\maxv \left(\mathcal{X}_\epsilon^+ - \widehat\psi\mathcal{X}_\epsilon\right)-2(1+\widehat\psi)\overline\epsilon$ or $\displaystyle\min_{j\in\mathcal{M}} \left(\mathcal{X}_\epsilon^+ [j]- \widehat\psi\mathcal{X}_\epsilon[j]\right) \le (1+\widehat\psi)\overline\epsilon$ \label{step_find2}\;
    Let $\psi_{\ell+1}=\widehat\psi$\;	
    
    Obtain $\displaystyle C_{A\ell}=C_{\epsilon}(\alpha=\frac{\psi_{\ell}+\psi_{\ell+1}}{2})$ \label{step_optstart2}\;	
    
    Compute $\mathcal{Q}_{A\ell}=\begin{bmatrix}C_{A\ell} & -\mathcal{P}_\epsilon\end{bmatrix} $ and $M_{A\ell}=I-\mathcal{Q}_{A\ell}\left(\mathcal{Q}_{A\ell}^{\top}\mathcal{Q}_{A\ell}\right)^{-1}\mathcal{Q}_{A\ell}^{\top}$ 
    \label{step_CpQ2}\;
		
		Solve $\displaystyle  \widehat{\alpha}_{A\ell}= \underset{{\alpha_{A\ell}\in(\psi_{\ell},\psi_{\ell+1})}}{\arg\minv} ~\frac{\|M_{A\ell}\left(\mathcal{X}_\epsilon^+-\alpha_{A\ell} \mathcal{X}_\epsilon\right)\|_2^2}{2}$. 
		
		Compute $\displaystyle{\mathcal{J}}(\widehat{\alpha}_{A\ell})=\frac{1}{2}{\|M_{A\ell}\left(\mathcal{X}_\epsilon^+-\widehat{\alpha}_{A\ell} \mathcal{X}_\epsilon\right)\|_2^2}$ \label{step_optend2} \;
		
		Obtain $\displaystyle C_{B\ell}=C_{\epsilon}(\alpha=\psi_{\ell+1})$ \label{step_cpstart2}\;	
		
		Compute $\mathcal{Q}_{B\ell}=\begin{bmatrix}C_{B\ell} & -\mathcal{P}_\epsilon\end{bmatrix} $ and $M_{B\ell}=I-\mathcal{Q}_{B\ell}\left(\mathcal{Q}_{B\ell}^{\top}\mathcal{Q}_{B\ell}\right)^{-1}\mathcal{Q}_{B\ell}^{\top}$\;
		
		Let $\alpha_{B\ell}=\psi_{\ell+1}$. Compute $\displaystyle{\mathcal{J}}(\widehat{\alpha}_{B\ell})=\frac{1}{2}{\|M_{B\ell}\left(\mathcal{X}_\epsilon^+-\widehat{\alpha}_{B\ell} \mathcal{X}_\epsilon\right)\|_2^2}$  \label{step_cpend2}\;
		
		Update $\mathcal{S}=\left\{i~\big|~ \left(\mathcal{X}_\epsilon^+[i] - \psi_{\ell+1}\mathcal{X}_\epsilon[i]\right)\le\maxv\left(\mathcal{X}_\epsilon^+ - \psi_{\ell+1}\mathcal{X}_\epsilon\right)-2(1+\psi_{\ell+1})\overline\epsilon,~ i\in\mathcal{T}\right\}$; 
		$\mathcal{Z}=\left\{i~\big|~ \left(\mathcal{X}_\epsilon^+[i] - \psi_{\ell+1}\mathcal{X}_\epsilon[i]\right)\le (1+\psi_{\ell+1})\overline\epsilon ,~ i\in\mathcal{T}\right\}$; $\mathcal{M}=\left\{i~\big|~i\not\in\mathcal{S}\bigcup\mathcal{Z} ,~ i\in\mathcal{T}\right\}$  \label{step_update2} 	\;
		
		$\ell=\ell+1$ \label{step_itend2}\;
	}
	\textbf{Output} $\displaystyle \widehat{\alpha}=  \underset{\alpha\in\{\widehat{\alpha}_{A\ell}\}\bigcup\{\widehat{\alpha}_{B\ell}\}}{\arg\minv} ~    \mathcal{J}(\alpha)$
\end{algorithm*}

\medskip
\begin{theorem}\longthmtitle{Properties of
    Algorithm~\ref{Algorithm_Main2}}\label{TH_alg2}
  Suppose Assumption \ref{AS_Unique} holds for $E_{\epsilon}(\alpha)$,
  $\mathcal{X}_{\epsilon}$ and $\mathcal{P}_{\epsilon}$. Further
  assume that the smallest eigenvalue of the matrix
  $\begin{bmatrix}
     \mathcal{X}_\epsilon
     &
       \mathcal{P}_\epsilon
   \end{bmatrix}^{\top}
   \left(I-E_\epsilon({\alpha}^{\star})^2\right)
   \left(I-E_\epsilon({\widehat\alpha})^2\right)
   \begin{bmatrix}
     \mathcal{X}_\epsilon
     &
       \mathcal{P}_\epsilon
   \end{bmatrix}$
   is lower bounded by $\lambda^2_{\min}>0$.
   Then, the Algorithm \ref{Algorithm_Main2} has the following properties:
   \begin{enumerate}[label=\textbf{\alph*}.]
   \item $\textit{[Minimizer]}$ The output $\widehat{\alpha}$ is the
     minimizer to problem \eqref{eq_ModifiedJ}.
   \item $\textit{[Complexity]}$ Algorithm \ref{Algorithm_Main2}
     terminates in at most $3nT_d+1$ number of iterations. The
     computational complexity of the algorithm is
     $\mathcal{O}(nT_d)^{3.34}$, where $n$ is the number of
     system nodes and $T_d$ is the number of sampled data.
   \item $\textit{[Identification]}$ Given $\widehat{\alpha}$
     obtained from Algorithm \ref{Algorithm_Main2}, one can
     compute
     $\widehat{\xi}=
     \begin{bmatrix}
       \widehat{v}^{\top}
       &
         \widehat{\bm{h}}^{\top}
     \end{bmatrix}^{\top}$ by solving \eqref{eq_compvh2}. 
     Let ${\alpha}^{\star}$ and ${\bm{h}}^{\star}$ be the true
     parameters of system \eqref{eq_dstmodel}. Then, there exists a
     constant $\zeta>0$ such that
     \begin{align}\label{eq_alphaherr}
       \left\|
       \begin{bmatrix}
         \widehat{\alpha}-\alpha^{\star}
         \\
         \widehat{\bm{h}}-\bm{h}^{\star}
       \end{bmatrix}
       \right\|_2
       \le \zeta~ \overline\epsilon
     \end{align}
     where $\overline\epsilon$ is defined in
     Assumption~\ref{AS_Err}. Furthermore, $s_D$ can be estimated by
     \begin{align}\label{eq_sdnoise}
       s_D=
       \frac{1}{|\mathcal{S}(\widehat{\alpha})|}
       \sum_{i\in\mathcal{S}(\widehat{\alpha})}\left(\mathcal{X}_\epsilon^+ 
       - \widehat{\alpha}\mathcal{X}_\epsilon\right)[i] 
     \end{align}
     where
     $\mathcal{S}(\widehat{\alpha}) \!=\!
     \Big\{i~\big|\left(\mathcal{X}_\epsilon^+[i] -
       \widehat{\alpha}\mathcal{X}_\epsilon[i]\right)\ge\maxv\left(\mathcal{X}_\epsilon^+
       - \widehat{\alpha}\mathcal{X}_\epsilon\right) -
     2(1+\widehat\alpha)\overline\epsilon,~ i\in\mathcal{T}\Big\}$ is
     the set of the entries of
     $\left(\mathcal{X}_\epsilon^+ -
       \widehat{\alpha}\mathcal{X}_\epsilon\right)$ that reach the
     upper saturation threshold.
   \end{enumerate}
 \end{theorem}
 
 
 \proof The argument to establish the properties of
 Algorithm~\ref{Algorithm_Main2} follows a similar path to that of
 Algorithm~\ref{Algorithm_Main} by introducing the critical points
 $\psi_{\ell}$ to determine the domains where $C_{\epsilon}(\alpha)$
 is constant, and then solve the piece-wise optimization problem on
 each domain. Our key efforts aims to prove statement \textbf{c}.

 \noindent
 \textbf{a}.
 The proof of Theorem \ref{TH_alg2} \textbf{a} is a direct
 generalization of Theorem \ref{TH_alg1} \textbf{a}, and is omitted
 for brevity.

\smallskip
\noindent
\textbf{b}. 
In Algorithm~\ref{Algorithm_Main2}, $\psi_{\ell}$ are still defined as
the points where the sets $\mathcal{S}$, $\mathcal{M}$, $\mathcal{Z}$
change their elements. But due to the modified definition for matrix
$C_{\epsilon}(\alpha)$, we have a slightly different way to determine
the values for $\psi_{\ell}$. As a consequence, the upper bound for
the number of $\psi_{\ell}$ changes from $2nT_d$ to $3nT_d$.  To
justify this change, consider the coordinate plane shown in
Fig. \ref{Fig_Alg2}, with $\alpha\in(0,\alpha_{\max}]$ being the
horizontal axis. We use the epigraph of
$\maxv \left(\mathcal{X}_\epsilon^+ -
  \alpha\mathcal{X}_\epsilon\right)- 2(1+\alpha)\overline\epsilon$ to
characterize the upper threshold of
$\left(\mathcal{X}_\epsilon^+-\alpha \mathcal{X}_\epsilon\right)[i]$;
and the hypograph of $(1+\alpha)\overline\epsilon$ to characterize the
lower threshold of
$\left(\mathcal{X}_\epsilon^+-\alpha \mathcal{X}_\epsilon\right)[i]$.
Note that a new $\psi_{\ell}$ is generated, when the lines
$\left(\mathcal{X}_\epsilon^+-\alpha \mathcal{X}_\epsilon\right)[i]$,
$i\in\until{nT_d}$ intersect with these two areas. For the epigraph of
$\maxv \left(\mathcal{X}_\epsilon^+ -
  \alpha\mathcal{X}_\epsilon\right)- 2(1+\alpha)\overline\epsilon$,
since it is convex, each line can intersect this area at most twice,
thus, for $i\in\until{nT_d}$ one has at most $2nT_d$ number of
$\psi_{\ell}$. For the hypograph of $(1+\alpha)\overline\epsilon$,
since $(1+\alpha)\overline\epsilon$ increases with $\alpha$ and
$\left(\mathcal{X}_\epsilon^+-\alpha \mathcal{X}_\epsilon\right)[i]$
decreases with $\alpha$, for $i\in\until{nT_d}$ one has at most $nT_d$
number of $\psi_{\ell}$. Furthermore, the algorithm starts with
$\psi_{0}=0$, which requires an extra round of execution. By putting
them together, Algorithm \ref{Algorithm_Main2} terminates in at most
$3nT_d+1$ number of iterations.

For the computational complexity, according to the proof of Theorem
\ref{TH_alg1} \textbf{b}, each step in Algorithm \ref{Algorithm_Main2}
has a complexity no larger than $\mathcal{O}(nT_d)^{2.34}$. Since
Algorithm \ref{Algorithm_Main2} requires no more than $3nT_d+1$ number
of iterations, its overall computational complexity is
$\mathcal{O}(nT_d)^{3.34}$.

\begin{figure}[h!]
  \vspace{-0.1cm}
  \centering
  \includegraphics[width=6.5 cm]{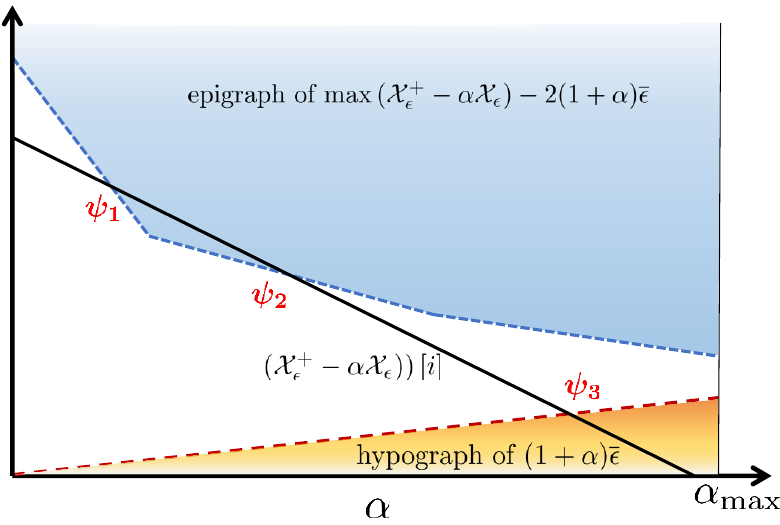}
  \caption{Each
    $\big(\mathcal{X}_\epsilon^+-\alpha \mathcal{X}_\epsilon\big)[i]$
    intersects at most twice with the epigraph of
    $\maxv \big(\mathcal{X}_\epsilon^+ -
    \alpha\mathcal{X}_\epsilon\big)- 2(1+\alpha)\overline\epsilon$ and
    one time with the hypograph of $(1+\alpha)\overline\epsilon$.}
  \label{Fig_Alg2}
\end{figure}

\smallskip
\noindent
\textbf{c}.  
From the system model, since ${\alpha}^{\star}$ and $\bm{h}^{\star}$
are associated with the true parameters of the system, there exists
${v}^{\star}$ such that
\begin{align}\label{eq_aplhavhstarnoise}
  \mathcal{X}^+-{\alpha}^{\star}
  \mathcal{X}+C_\epsilon({\alpha}^{\star}){v}^{\star}-\mathcal{P}{\bm{h}^{\star}}=0,
\end{align}
where $\mathcal{X}^+$, $\mathcal{X}$, $\mathcal{P}$ are data without
measurement noise.  By the definitions of $\mathcal{X}_\epsilon^+$,
$\mathcal{X}_\epsilon$, and $\mathcal{P}_\epsilon$, there holds
\begin{align}\label{eq_noiseequality}
  \left(\mathcal{X}_\epsilon^+\!-\!\epsilon_{\mathcal{X}^+}\right)\!-\!{\alpha}^{\star}
  \left(\mathcal{X}_\epsilon-\epsilon_{\mathcal{X}}\right)\!+\!C_\epsilon({\alpha}^{\star}){v}^{\star}\!-\!\left(\mathcal{P}_\epsilon-\epsilon_{\mathcal{P}}\right){\bm{h}^{\star}}=0 ,
\end{align}
where $\epsilon_{\mathcal{P}}=\mathcal{P}-\mathcal{P}_\epsilon$, and
from Assumption \ref{AS_Err}, we know
$|\epsilon_{\mathcal{P}}|_{\infty}\le\overline\epsilon$.
Equation \eqref{eq_noiseequality} yields
\begin{align}\label{eq_XXQxi}
  \mathcal{X}_\epsilon^+ -{\alpha}^{\star} \mathcal{X}_\epsilon
  &=\epsilon_{\mathcal{X}^+} -{\alpha}^{\star}\epsilon_{\mathcal{X}}-C_\epsilon({\alpha}^{\star}){v}^{\star}+\left(\mathcal{P}_\epsilon-\epsilon_{\mathcal{P}}\right){\bm{h}^{\star}}\nonumber\\
  &=-Q_{\epsilon}({\alpha}^{\star})\xi^{\star}+\Gamma(\bm{\epsilon})
\end{align}
where
$\Gamma(\bm{\epsilon})\triangleq\epsilon_{\mathcal{X}^+}-
{\alpha}^{\star}\epsilon_{\mathcal{X}}-\epsilon_{\mathcal{P}}{\bm{h}^{\star}}$
and
$\bm{\epsilon}=\col\{\epsilon_{\mathcal{X}^+}, \epsilon_{\mathcal{X}},
\epsilon_{\mathcal{P}}\}$.  Now, since
$M_{\epsilon}(\alpha^{\star})Q_{\epsilon}({\alpha}^{\star})=0$ and
$\|M_{\epsilon}(\alpha^{\star})\|_2\le1$, one has
\begin{align}\label{eq_Jastar}
  \mathcal{J}(\alpha^{\star})&=\frac{\|M_{\epsilon}(\alpha^{\star})(\mathcal{X}_{\epsilon}^+-\alpha^{\star}\mathcal{X}_{\epsilon})\|_2^2}{2}\nonumber\\
                             &=\frac{\|M_{\epsilon}(\alpha^{\star})\Gamma(\bm{\epsilon})\|_2^2}{2}\le\frac{\|\Gamma(\bm{\epsilon})\|_2^2}{2}.
\end{align}
Recall that $\widehat{\alpha}$ is the minimizer to problem
\eqref{eq_ModifiedJ} and thus
\begin{align}\label{eq_JleJ}
  \mathcal{J}(\widehat\alpha)\le \mathcal{J}(\alpha^{\star})\le\frac{\|\Gamma(\bm{\epsilon})\|_2^2}{2}.
\end{align}
Since $I-E_{\epsilon}({\alpha}^{\star})^2$ and
$I-E_{\epsilon}({\widehat\alpha})^2$ are diagonal matrices with only
$1$ or $0$ entries, one has
\begin{align}\label{eq_IEIEnorm}
  &\frac{\|\left(I-E_{\epsilon}({\alpha}^{\star})^2\right)\left(I-E_{\epsilon}({\widehat\alpha})^2\right)M_{\epsilon}(\widehat\alpha)(\mathcal{X}_{\epsilon}^+\!-\widehat\alpha\mathcal{X}_{\epsilon})\|_2^2}{2}\nonumber\\
  &~~~\le  \frac{\|M_{\epsilon}(\widehat\alpha)(\mathcal{X}_{\epsilon}^+\!-\widehat\alpha\mathcal{X}_{\epsilon})\|_2^2}{2}
    =\mathcal{J}(\widehat\alpha)\le \frac{\|\Gamma(\bm{\epsilon})\|_2^2}{2}.
\end{align}
Define
$\theta_\epsilon(\alpha)=\left(\mathcal{Q}_\epsilon(\alpha)^{\top}
  \mathcal{Q}_\epsilon(\alpha)\right)^{-1}\mathcal{Q}_\epsilon(\alpha)^{\top}\left(\mathcal{X}_\epsilon^+-\alpha
  \mathcal{X}_\epsilon\right)$.  From the definition of
$M_{\epsilon}(\alpha)$,
\begin{align}\label{eq_MXAX}
  &M_{\epsilon}(\widehat\alpha)(\mathcal{X}_{\epsilon}^+-\widehat\alpha\mathcal{X}_{\epsilon})=(\mathcal{X}_{\epsilon}^+-\widehat\alpha\mathcal{X}_{\epsilon})-Q_{\epsilon}(\widehat\alpha)\theta_{\epsilon}(\widehat\alpha)\nonumber\\
  &=(\alpha^{\star}-\widehat\alpha)\mathcal{X}_{\epsilon}+Q_{\epsilon}(\widehat\alpha)\theta_{\epsilon}(\widehat\alpha)-Q_{\epsilon}({\alpha}^{\star})\xi^{\star}+\Gamma(\bm{\epsilon}) .
\end{align}
The last equality is obtained by introducing equation
\eqref{eq_XXQxi}. Now, due to the fact that
$\left(I-E_{\epsilon}({\alpha})^2\right)\!C_{\epsilon}({\alpha}) =
\left(I-C_{\epsilon}({\alpha})C_{\epsilon}({\alpha})^{\top}\right)C_{\epsilon}({\alpha})
= C_{\epsilon}({\alpha})-C_{\epsilon}({\alpha})I=0$, similar to the
derivation in \eqref{eq_EQT}, one has
\begin{align}\label{eq_IEIEMXAX}
  &\left(I-E_{\epsilon}({\alpha}^{\star})^2\right)\left(I-E_{\epsilon}({\widehat\alpha})^2\right)M_{\epsilon}(\widehat\alpha)(\mathcal{X}_{\epsilon}^+\!-\widehat\alpha\mathcal{X}_{\epsilon})\nonumber\\
  =&\left(I-E_{\epsilon}({\alpha}^{\star})^2\right)\left(I-E_{\epsilon}({\widehat\alpha})^2\right)\Big(\mathcal{X}_{\epsilon}(\alpha^{\star}-\widehat\alpha)+\mathcal{P}_{\epsilon}\widetilde\theta_{\epsilon}(\widehat\alpha)\nonumber\\
  &~~~~~~\qquad\qquad\qquad\qquad\qquad\qquad-\mathcal{P}_{\epsilon}\widetilde\xi^{\star}+\Gamma(\bm{\epsilon})\Big)\nonumber\\
  =&\left(I-E_{\epsilon}({\alpha}^{\star})^2\right)\left(I-E_{\epsilon}({\widehat\alpha})^2\right)
     \left(\begin{bmatrix}
         \mathcal{X}_{\epsilon}& \mathcal{P}_{\epsilon}
	\end{bmatrix}\begin{bmatrix}
          \alpha^{\star}-\widehat\alpha\\
          \widetilde
          \theta_{\epsilon}(\widehat\alpha)-\widetilde\xi^{\star}
	\end{bmatrix}\right)\nonumber\\
	&~~~~+\left(I-E_{\epsilon}({\alpha}^{\star})^2\right)\left(I-E_{\epsilon}({\widehat\alpha})^2\right)\Gamma(\bm{\epsilon}) ,
\end{align}
where
$\widetilde\theta_{\epsilon}(\widehat\alpha)
=
\begin{bmatrix}
  \bm{0}_{d(\widehat\alpha)\times n(n+m-1)} & I_{n(n+m-1)}
\end{bmatrix} \theta_{\epsilon}(\widehat\alpha)$ and
$\widetilde\xi^{\star}=
\begin{bmatrix}
  \bm{0}_{d(\alpha^{\star})\times n(n+m-1)} &  I_{n(n+m-1)}
\end{bmatrix} \xi^{\star}$.
Bringing this back to \eqref{eq_IEIEnorm}, since $\|\left(I-E_{\epsilon}({\alpha}^{\star})^2\right)\left(I-E_{\epsilon}({\widehat\alpha})^2\right)\|_2\le1$, there holds 
\begin{align}\label{eq_IEIEXP}
  &\left\|\left(I-E_{\epsilon}({\alpha}^{\star})^2\right)\left(I-E_{\epsilon}({\widehat\alpha})^2\right)
    \begin{bmatrix}
      \mathcal{X}_{\epsilon}
      & \mathcal{P}_{\epsilon}
    \end{bmatrix}
    \begin{bmatrix}
      \alpha^{\star}-\widehat\alpha\\
      \widetilde\theta_{\epsilon}(\widehat\alpha)-\widetilde\xi^{\star}
    \end{bmatrix}\right\|_2^2\nonumber\\
  &\qquad\qquad\qquad\qquad\qquad\qquad\qquad\qquad\le 2 \|\Gamma(\bm{\epsilon})\|_2^2.
\end{align}

In \eqref{eq_compvh2},
$\widehat{\xi}=
\begin{bmatrix}
  \widehat{v}^{\top}
  &
    \widehat{\bm{h}}^{\top}
\end{bmatrix}^{\top}$
is obtained based on $\alpha=\widehat{\alpha}$. Note that,
$\widehat\alpha$, $\widehat{\xi}$ is not a set of minimizers to
$\mathcal{J}_0(\alpha,{\xi})$ in \eqref{eq_obj2bn}, as
$\widehat\alpha$ is pre-determined by minimizing
$\mathcal{J}(\alpha)$ in \eqref{eq_ModifiedJ}. However, given
$\widehat\alpha$ is fixed, since \eqref{eq_compvh2} minimizes
$\mathcal{J}_0(\alpha=\widehat{\alpha},{\xi})$ w.r.t. $\xi$, we have 
$\mathcal{J}_0(\widehat\alpha,\widehat{\xi})\le
\mathcal{J}_0(\widehat\alpha,\xi)$,
for any vector $\xi$ that satisfies the constraint
$S_{\epsilon}\xi\le0$.  Now, consider
\begin{align*}
  \overline\xi=\begin{bmatrix}
    I_{d(\widehat\alpha)} &  \\
    &  \bm{0}_{n(n+m-1)}
  \end{bmatrix} \theta_{\epsilon}(\widehat\alpha) + \begin{bmatrix}
    \bm{0}_{d(\alpha)^{\star}} &  \\
    &  I_{n(n+m-1)}
  \end{bmatrix} \xi^{\star}.
\end{align*}
Given $
S_{\epsilon}=\begin{bmatrix}
  \bm{0} & D
\end{bmatrix}
$, one has $S_{\epsilon} \overline\xi=S_{\epsilon}\xi^{\star}$. Since
$\xi^{\star}$ is the true parameter of the model,
$S_{\epsilon}\xi^{\star}\le0$ inherently holds. Thus,
$S_{\epsilon}\overline\xi\le0$ and
\begin{align}\label{eq_J0J0}
  \mathcal{J}_0(\widehat\alpha,\widehat{\xi})\le \mathcal{J}_0(\widehat\alpha,\overline\xi).
\end{align}

For the right-hand side of \eqref{eq_J0J0}, by definition,
\begin{align*}
  \mathcal{J}_0(\widehat\alpha,\overline\xi)=\frac{1}{2}\left\|\mathcal{X}_\epsilon^+-\widehat{\alpha}
  \mathcal{X}_\epsilon+Q_\epsilon(\widehat{\alpha})\overline\xi\right\|_2^2 
\end{align*}
where
\begin{align*}
  &\mathcal{X}_\epsilon^+-\widehat{\alpha}
    \mathcal{X}_\epsilon+Q_\epsilon(\widehat{\alpha})\overline\xi
  \\
  =\
  & \left(\mathcal{X}_\epsilon^+-\widehat{\alpha}
    \mathcal{X}_\epsilon+Q_\epsilon(\widehat{\alpha}){\theta_\epsilon(\widehat\alpha)}
    \right)  
    +Q_\epsilon(\widehat{\alpha})\left(\overline\xi -
    {\theta_\epsilon(\widehat\alpha)}  \right). 
\end{align*}
Thus,
\begin{align}\label{eq_J0JQxitheta}
  \mathcal{J}_0(\widehat\alpha,\overline\xi)\le 2\mathcal{J}(\widehat\alpha)+{
  \| Q_\epsilon(\widehat{\alpha})\left(\overline\xi - {\theta_\epsilon(\widehat\alpha)}  \right) \|_2^2  }
\end{align}
where
$2\mathcal{J}(\widehat\alpha) \le \|\Gamma(\bm{\epsilon})\|_2^2$.
Furthermore, by  $\widetilde\theta_{\epsilon}(\widehat\alpha)$ and
$\widetilde\xi^{\star}$ in \eqref{eq_IEIEMXAX}, there holds
\begin{align*}
  &Q_\epsilon(\widehat{\alpha})\left(\overline\xi - {\theta_\epsilon(\widehat\alpha)}  \right)\\
  =\ &Q_\epsilon(\widehat{\alpha})\!\left(\begin{bmatrix}
      \bm{0}_{d(\widehat\alpha)} &  \\
      &  \!\!\!\!\!I_{n(n+m-1)}
    \end{bmatrix} \!\theta_{\epsilon}(\widehat\alpha) \!+\! \begin{bmatrix}
      \bm{0}_{d(\alpha^{\star})} &  \\
      & \!\!\!\!\!\! I_{n(n+m-1)}
    \end{bmatrix}\! \xi^{\star} \! \right)\\
  =\ &
       \mathcal{P}_\epsilon
       \left(\widetilde\theta_{\epsilon}(\widehat\alpha)-\widetilde\xi^{\star}\right) .
\end{align*}
From \eqref{eq_IEIEXP}, we know
$\left\|\widetilde\theta_{\epsilon}(\widehat\alpha)-\widetilde\xi^{\star}\right\|_2^2$
is linearly bounded by $\|\Gamma(\bm{\epsilon})\|_2^2$. Thus, there
must exist $\eta>0$ such that
\begin{align}\label{eq_JAXIH}
  \mathcal{J}_0(\widehat\alpha,\overline\xi) \le
  \eta\|\Gamma(\bm{\epsilon})\|_2^2. 
\end{align}

For the left-hand side of \eqref{eq_J0J0}, by definition,
\begin{align}\label{eq_JnoiseM}
  \mathcal{J}_0(\widehat\alpha,\widehat{\xi})=\frac{1}{2}\left\|\mathcal{X}_\epsilon^+-{\widehat\alpha}
  \mathcal{X}_\epsilon+C_\epsilon({\widehat\alpha})\widehat{v}-\mathcal{P}_\epsilon{\widehat{\bm{h}}}\right\|_2^2 .
\end{align}	
By introducing equation \eqref{eq_noiseequality} and
$\Gamma(\bm{\epsilon})=\epsilon_{\mathcal{X}^+}-
{\alpha}^{\star}\epsilon_{\mathcal{X}}-\epsilon_{\mathcal{P}}{\bm{h}^{\star}}$
into \eqref{eq_JnoiseM}, there holds
\begin{align*}	
  \mathcal{J}_0(\widehat\alpha,\widehat{\xi})\!=\!\frac{1}{2}\!\left\|
  \text{\footnotesize $\begin{bmatrix}
                         -\mathcal{X}_\epsilon & \!\!\!-\mathcal{P}_\epsilon
    \end{bmatrix}\!\begin{bmatrix}
      \widehat\alpha\!-\!\alpha^{\star}\\
      \widehat{\bm{h}}\!-\!\bm{h}^{\star}
    \end{bmatrix}
  \!-\!C_\epsilon({\alpha}^{\star}){v}^{\star}
  \!\!+\!C_\epsilon({\alpha}){v} \!+\! \Gamma(\bm{\epsilon})$}
  \right\|_2^2 .
\end{align*}
Recalling the facts that $\|I-E_\epsilon({\alpha})^2\|_2^2\le1$,
$I-E_\epsilon({\alpha})^2$ is diagonal, and
$\left(I-E_\epsilon({\alpha})^2\right)C_\epsilon({\alpha})=0$, then
\begin{align}\label{eq_Jnoise2M}
  &\mathcal{J}_0(\widehat\alpha,\widehat{\xi})
    \ge 
    \text{\footnotesize$
    \frac{1}{2}\Big\|\left(I-E_\epsilon({\alpha}^{\star})^2\right)\left(I-E_\epsilon({\widehat\alpha})^2\right)$}\nonumber\\	&~~~~~~~~~~~~\text{\footnotesize $\left(\begin{bmatrix}
        -\mathcal{X}_\epsilon & \!\!\!-\mathcal{P}_\epsilon
      \end{bmatrix}\!\begin{bmatrix}
        \widehat\alpha\!-\!\alpha^{\star}\\
        \widehat{\bm{h}}\!-\!\bm{h}^{\star}
      \end{bmatrix}
  \!-\!C_\epsilon({\alpha}^{\star}){v}^{\star} \!\!+\!C_\epsilon({\alpha}){v} \!+\! \Gamma(\bm{\epsilon})\right)\Big\|_2^2$}\nonumber\\
  =&
     \text{\footnotesize $\frac{1}{2}\left\|\left(I-E_\epsilon({\alpha}^{\star})^2\right)\!\!\left(I-E_\epsilon({\widehat\alpha})^2\right)\!\!\left(\!\begin{bmatrix}
             -\mathcal{X}_\epsilon & \!\!-\mathcal{P}_\epsilon
           \end{bmatrix}\!\begin{bmatrix}
             \widehat\alpha\!-\!\alpha^{\star}\\
             \widehat{\bm{h}}\!-\!\bm{h}^{\star}
           \end{bmatrix}\!+\!\Gamma(\bm{\epsilon})\right)\right\|_2^2$} .
\end{align}
The last equality holds because
$\left(I-E_\epsilon({\alpha})^2\right)C_\epsilon({\alpha}) =
\left(I-C_\epsilon({\alpha})C_\epsilon({\alpha})^{\top}\right)C_\epsilon({\alpha})
= C_\epsilon({\alpha})-C_\epsilon({\alpha})I=0$.

Now, by combining equations \eqref{eq_J0J0}, \eqref{eq_JAXIH} and
\eqref{eq_Jnoise2M}, one has
\begin{align}\label{eq_noise_neq1M}
  \left\|\left(I-E_\epsilon({\alpha}^{\star})^2\right)\!\!\left(I-E_\epsilon({\widehat\alpha})^2\right)\!\begin{bmatrix}
      -\mathcal{X}_\epsilon & \!\!-\mathcal{P}_\epsilon
    \end{bmatrix}\!
                              \begin{bmatrix}
      \widehat\alpha\!-\!\alpha^{\star}
      \\
      \widehat{\bm{h}}\!-\!\bm{h}^{\star}
    \end{bmatrix}\right\|_2^2\nonumber
  \\
  \le2\left(1+\eta \right)\left\|\Gamma(\bm{\epsilon})\right\|_2^2 .
\end{align}
Since we have assumed that the smallest eigenvalue of the matrix
$\begin{bmatrix} \mathcal{X}_\epsilon & \mathcal{P}_\epsilon
\end{bmatrix}^{\top}\left(I-E_\epsilon({\alpha}^{\star})^2\right)\left(I-E_\epsilon({\widehat\alpha})^2\right)\begin{bmatrix}
  \mathcal{X}_\epsilon & \mathcal{P}_\epsilon
\end{bmatrix}$ is lower bounded by $\lambda^2_{\min}>0$, then
\begin{align}
  \left\| \begin{bmatrix}
      \widehat\alpha-\alpha^{\star}\\
      \widehat{\bm{h}}-\bm{h}^{\star}
    \end{bmatrix} \right\|_2^2\le
  \frac{2}{\lambda^2_{\min}}\left(1+\eta
  \right)\left\|\Gamma(\bm{\epsilon})\right\|_2^2.
\end{align}
Since $\|\bm{\epsilon}\|_{\infty}\le\overline\epsilon$, and the
dimension of $\epsilon$ is finite, there must exist a constant
$\zeta>0$ such that for all $\overline\epsilon\ge0$,
\begin{align}
  \left\|\begin{bmatrix}
      \widehat{\alpha}-\alpha^{\star}\\
      \widehat{\bm{h}}-\bm{h}^{\star}
    \end{bmatrix}\right\|_2\le \zeta~ \overline\epsilon
\end{align}

Finally, for $s_D$, compared with Algorithm \ref{Algorithm_Main},
where we directly use
$s_D=\maxv\left(\mathcal{X}^+ - \alpha^{\star}\mathcal{X}\right)$, in
the presence of noise, we only know that
$\maxv\left(\mathcal{X}^+ -
  \alpha^{\star}\mathcal{X}\right)-2(1+\alpha^{\star})\overline\epsilon\le
s_D\le\maxv\left(\mathcal{X}^+ -
  \alpha^{\star}\mathcal{X}\right)+2(1+\alpha^{\star})\overline\epsilon$.
Thus, in \eqref{eq_sdnoise}, we use the average of all the entries
that reach the upper saturation threshold to estimate $s_D$.
\hfill\qed

\begin{remark}\longthmtitle{Size of the data and computational
    complexity, continued.}
  In the presence of random measurement noise, the matrix
  $E_{\epsilon}(\alpha)$ in \eqref{eq_defcEnoise} has more $\pm1$
  entries (due to relaxation) than $E(\alpha)$ in
  \eqref{eq_defcE}. Thus, it requires a larger $T_d$ to make
  Assumption \ref{AS_Unique} hold. Furthermore, since the result of
  Algorithm \ref{Algorithm_Main2} is an approximation to the true
  parameter of the system, in general, adding more data sets may also
  be beneficial for obtaining more accurate system parameters. In
  future work, we plan to study the trade off between computational
  complexity and estimation accuracy. \hfill$\oldsquare$
\end{remark}

\section{Examples}
We present simulation results here to validate the effectiveness of the
proposed results.
%
%

\subsection{Simulation with synthetic data}
We consider a network with $n=10$ nodes. The dimension of the input
$\bm{u}$ is $m=10$. Given the state/input dimensions of the system, we
first create matrices $W_D\in\mathbb{R}^{10\times 10}$ and
$B_D\in\mathbb{R}^{10\times 10}$.  By definition, $W_D$ is a matrix
with $0$ diagonal entries.  For the non-zero entries of $W_D$, we make
sure they are consistent with \textit{Dale's law}, cf.
Remark~\ref{Rm_Dale}, i.e., each column of $W_D$ is either
non-negative or non-positive depending on the excitatory or inhibitory
properties of the nodes. The values of these entries are randomly
chosen from $[0 ~~0.1]$ or $[-0.05~~0]$ with uniform distributions.
For $B_D\in\mathbb{R}^{10\times 10}$, all its entries are randomly
chosen from $[-0.04 ~~0.06]$ with uniform distributions. We set
$\alpha^{\star}=0.9$ and $s_D=2$.  Based on $W_D$, $B_D$, $\alpha$ and
$s_D$, we create data samples, for $k\in \{1,\dots,T_d\}$ and
$T_d=250$. In this simulation, for different $k$, $\bm{x}_d(k)$ and
$\bm{u}_d(k)$ are chosen independently, i.e., the entries of
$\bm{x}_d(k)$ are randomly chosen from $[0 ~~4]$; the entries of
$\bm{u}_d(k)$ are randomly chosen from $[0 ~~6]$, with uniform
distributions. For each pair of $\bm{x}_d(k)$ and $\bm{u}_d(k)$, we
compute $\bm{x}_d^+(k)$ based on the discrete-time system
model~\eqref{eq_dataF}. It is worth pointing out that the obtained
data set satisfies Assumption \ref{AS_Unique} for all
$\alpha\in(0,1)$.

%

%
%


\subsubsection{Parameter identification with Algorithm
  \ref{Algorithm_Main2} under measurement noise}
To simulate the impact of measurement noise, we introduce
$\bm{x}^+_\epsilon(k)=\bm{x}^+_d(k)+\epsilon_{x^+}(k)$,
$ \bm{x}_\epsilon(k)=\bm{x}_d(k)+\epsilon_x(k)$, and
$\bm{u}_\epsilon(k)=\bm{u}_d(k)+\epsilon_u(k)$ where
$\epsilon_{x^+}(k), \epsilon_x(k), \epsilon_u(k)$ are the noises and
they satisfy $\|\epsilon_{x^+}(k)\|_{\infty}\le\overline\epsilon$,
$\|\epsilon_{x}(k)\|_{\infty}\le\overline\epsilon$, and
$\|\epsilon_{u}(k)\|_{\infty}\le \overline\epsilon$. Here,
$\overline\epsilon=0.1$.  By running Algorithm \ref{Algorithm_Main2},
we obtain $\alpha_{\max}=1$, and the function value of
$\mathcal{J}(\alpha)$ is $4.4428$ at $\widehat{\alpha}=0.9012$. The
estimation error for $\alpha$ is $0.0012$.  Given this
$\widehat{\alpha}$, one can obtain $\widehat{\bm{h}}$
from~\eqref{eq_compvh2}, then decode it into matrices $W_D$ and $B_D$
via~\eqref{eq_defzHP}.  To compare $\widehat{\bm{h}}$ with the true
$\bm{h}^{\star}$ of the system, the Root Mean Square Error (RMSE) of
$\bm{h}$ is,
$$\text{RMSE}(\bm{h})=\sqrt{\frac{\|\widehat{\bm{h}}-\bm{h}^{\star}\|}{n(n+m-1)}}=0.0039.$$
Finally, we identify
$$s_{_D}=
\frac{1}{|\mathcal{S}(\widehat{\alpha})|}\sum_{i\in\mathcal{S}(\widehat{\alpha})}\left(\mathcal{X}_\epsilon^+
  - \widehat{\alpha}\mathcal{X}_\epsilon\right)[i]=1.989.$$

\noindent\begin{figure}[h!]
  \vspace{-0.1cm}
  \centering
  \includegraphics[width=7.8 cm]{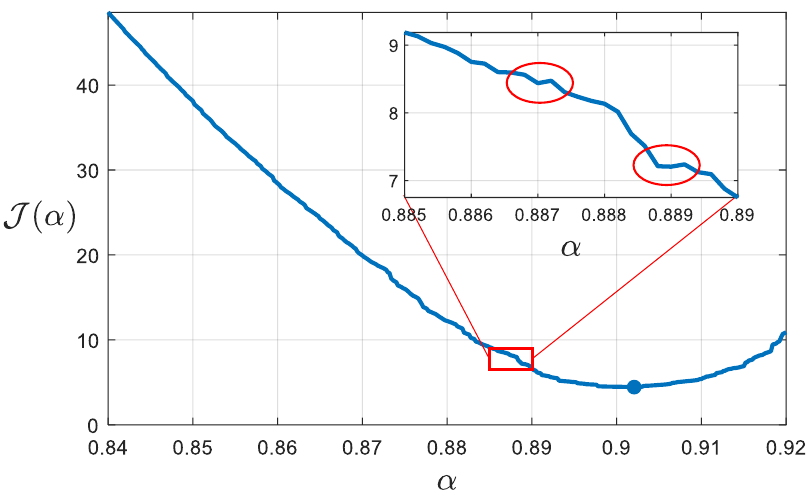}
  \caption{Identification of the system parameter $\alpha$ on a
    10-node network, with measurement noise bounded by
    $\overline\epsilon=0.1$.}
  \label{Fig_With_Noise}
\end{figure}

Fig. \ref{Fig_With_Noise} shows how $\mathcal{J}(\alpha)$ changes as a
function of~$\alpha$. Although the function appears roughly convex, the
magnified area reveals it to be non-smooth and
non-convex. Using gradient descent methods to solve for $\alpha$
can easily result in getting trapped at a local minimum.

\subsubsection{Comparing Algorithms~\ref{Algorithm_Main} and~\ref{Algorithm_Main2}}
Here, we show that Algorithm~\ref{Algorithm_Main2} outperforms
Algorithm~\ref{Algorithm_Main} in terms of estimation accuracy when
the data is subject to measurement noise. We use the same system model
introduced above but generate measurement noise with different
magnitudes, for
$\overline\epsilon\in\{0.02,~0.04,~0.06,~0.08,~0.1\}$. For each
$\overline\epsilon$, the noise is randomly generated for $70$ times
and the parameter is identified for each generated data. This allows
us to statistically analyze the estimation errors of the two
algorithms, as shown in Fig. \ref{Fig_Modified}.  Compared with
Algorithm~\ref{Algorithm_Main}, the advantage of
Algorithm~\ref{Algorithm_Main2} is remarkable when $\overline\epsilon$
is small; and the gap decreases as $\overline\epsilon$ goes large.

\noindent\begin{figure}[h!]
  \vspace{-0.1cm}
  \centering
  \includegraphics[width=9 cm]{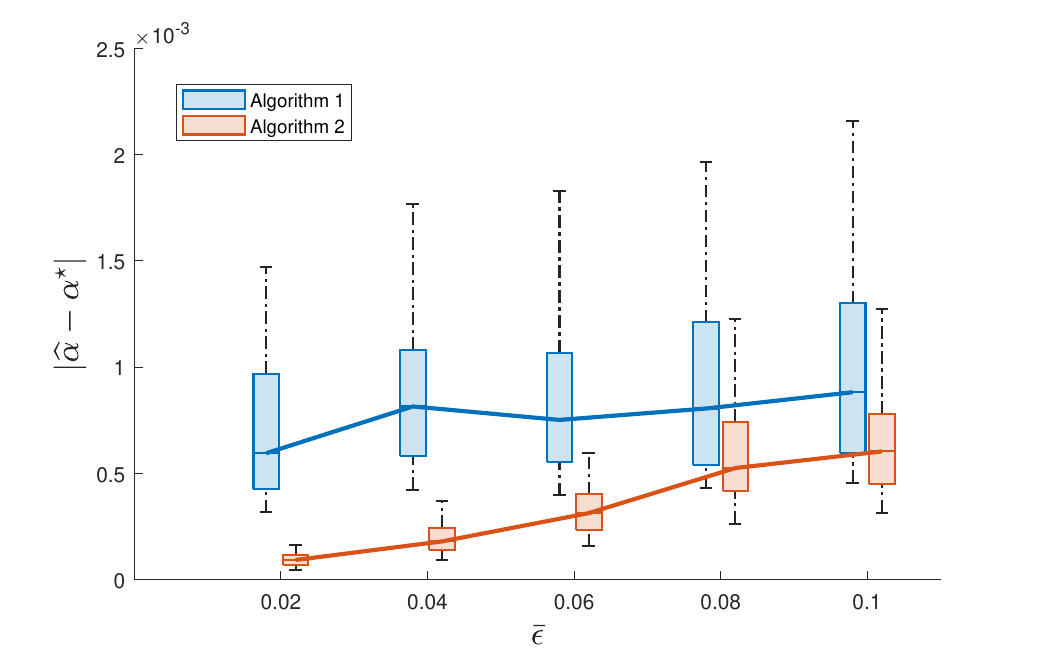}
  \caption{A box-plot with whiskers that compares the estimation
    errors of Algorithm \ref{Algorithm_Main} and Algorithm
    \ref{Algorithm_Main2}. The bottom and top of each box are the 25\%
    and 75\% of the samples, respectively. For each algorithm, the
    solid line connects the medians of the estimation errors.}
  \label{Fig_Modified}
\end{figure}

\subsubsection{Comparison with general nonlinear optimization solver}
The proposed algorithms are based on the reformulation in
Section~\ref{Sec_Scalar} that simplifies an optimization problem with
a large number of variables into a scalar optimization
\eqref{eq_reducedJ}. To demonstrate the advantage of such
reformulation, in terms of both accuracy and computational complexity,
we take $\overline\epsilon=0.04$ and compare in Table
\ref{Tab_compare} the proposed algorithms (ALG1,2) with two nonlinear
optimization solvers (NOS1,2) based on interior-point
methods~\cite{KAM-CLM-DS:89}. In particular, NOS1 aims to directly
minimize \eqref{eq_obj2b} with variables $\alpha, \bm{h}$. Although
the execution time is reasonable, the NOS1 is very unstable and
usually converges to a local minimum with the value of the objective
function larger than $100$. The obtained parameters also have large
errors ($\text{RMSE}(\bm{h})>1$).  The difficulty of solving NOS1 may
come from the fact that $\mathcal{Q}(\alpha)$ is a complex nonlinear
function of $\alpha$. To address this issue, we simplify the problem
by considering NOS2 with an objective function
$\frac{1}{2}\|\mathcal{X}^+-\alpha
\mathcal{X}-\left[\mathcal{P}\bm{h}\right]_0^{2}\|_2^2$ with variables
$\alpha,\bm{h}$. This function characterizes the mismatch of equation
\eqref{eq_compactdataeq} by assuming $s_{_D}=2$ is known. While NOS2
shows a clear performance improvement over NOS1, it takes a
significant amount of execution time due to the large number of
variables. In addition, the NOS2 still have a large performance gap
compared with the algorithms proposed here. Finally, the results for
Algorithms~\ref{Algorithm_Main} and~\ref{Algorithm_Main2} in Table
\ref{Tab_compare} match the ones in Fig. \ref{Fig_Modified}. The
number of $\psi_{\ell}$ are smaller than the bounds identified in
Theorem \ref{TH_alg2} \textbf{b}.

{
  \begin{table}[h]
    \small
    \renewcommand{\arraystretch}{1.3} 
    \centering
    \caption{Comparing the proposed algorithms with nonlinear optimization}
    \begin{tabular}{l|cccc}   \toprule 
      & ALG1  &ALG2&NOS1&NOS2\\ 
      \midrule
      execution & $<10s$ &$\sim250s$&$>100s$&$>1000s$\\
      $|\widehat{\alpha}-\alpha^{\star}|$ &$\sim8\times10^{-4}$ &$\sim2\times 10^{-4}$& $>0.1$ & $>10^{-3}$ \\
      Obj. value & $\sim20$    & $\sim4$  &$>100$  &$\sim40$\\
      $\text{RMSE}(\bm{h})$ & $\sim6\times10^{-2}$ & $\sim4\times10^{-3}$ &$>1$& $>0.1$ \\
      $\#$ of $\psi_{\ell}$ & $46$ &$1256$&$-$ & $-$\\ 
      \bottomrule \label{Tab_compare}
    \end{tabular}\label{Tab_obj}
    {\footnotesize Testing platform uses MATLAB with intel Core
      i9-9900kf CPU and 32 GB of RAM. The NOS1/2 employs
      \texttt{fmincon} solver with interior-point option.} 
\end{table}}

\subsection{Reconstruction of firing rate dynamics in rodents' brain}
We also apply the proposed algorithms to a real-world example on
goal-driven attention. The data we use is from a carefully designed
experimental paradigm \cite{CCR-MRD:14,CCR-MRD:14-crcns} that involves
selective listening in rodents. During the experiment, the rodents are
subject simultaneously to a (left/right) white noise burst and a
(high/low pitch) narrow-band warble. Which of the two sounds is
relevant and which is a distraction depends on the ``rule'' of the
trial.  During the experiments, the firing rates of the neuron cells
are recorded from two different regions of their brains, i.e.,
Prefrontal Cortex (PFC) and Primary Auditory Cortex (A1).  By using
the classification method introduced in \cite{EN-JC:21-tacII}, cf.
Fig.~\ref{Fig_neurons}, we classify all the neuron cells into $2^3=8$
groups based on a combination of the following properties: region
(PFC, A1); type (excitatory, inhibitory); and encoding (task relevant,
irrelevant). Then we consider each class of neurons as a node of the
system, and we assume the obtained network has the following
properties:
\begin{itemize}
\item The nodes are interconnected, the excitatory nodes have positive
  outgoing edges; and inhibitory nodes have negative outgoing
  edges. This is consistent with Dale's law~\cite{JCE:86}.
\item Excitatory nodes may have self-loops; inhibitory nodes do not
  have self-loops \cite{PS-GT-RL-EHB:98}.
\item The excitatory neurons in PFC and the inhibitory neurons in A1
  share similar time constants
  \cite{JDM-AB-DJF-RR-JDW-XC-CP-TP-HS-DL-XW:14}.
\end{itemize}
We use the average firing rate of the populated neurons as the state
of the node. The sampling duration in our example is 14 seconds, for
$t\in[-7 , 7]$ with an interval $\delta t=0.1s$. The stimuli happens
at time $t=0$.
\begin{figure}[h!]
  \centering
  \includegraphics[width=8.5 cm]{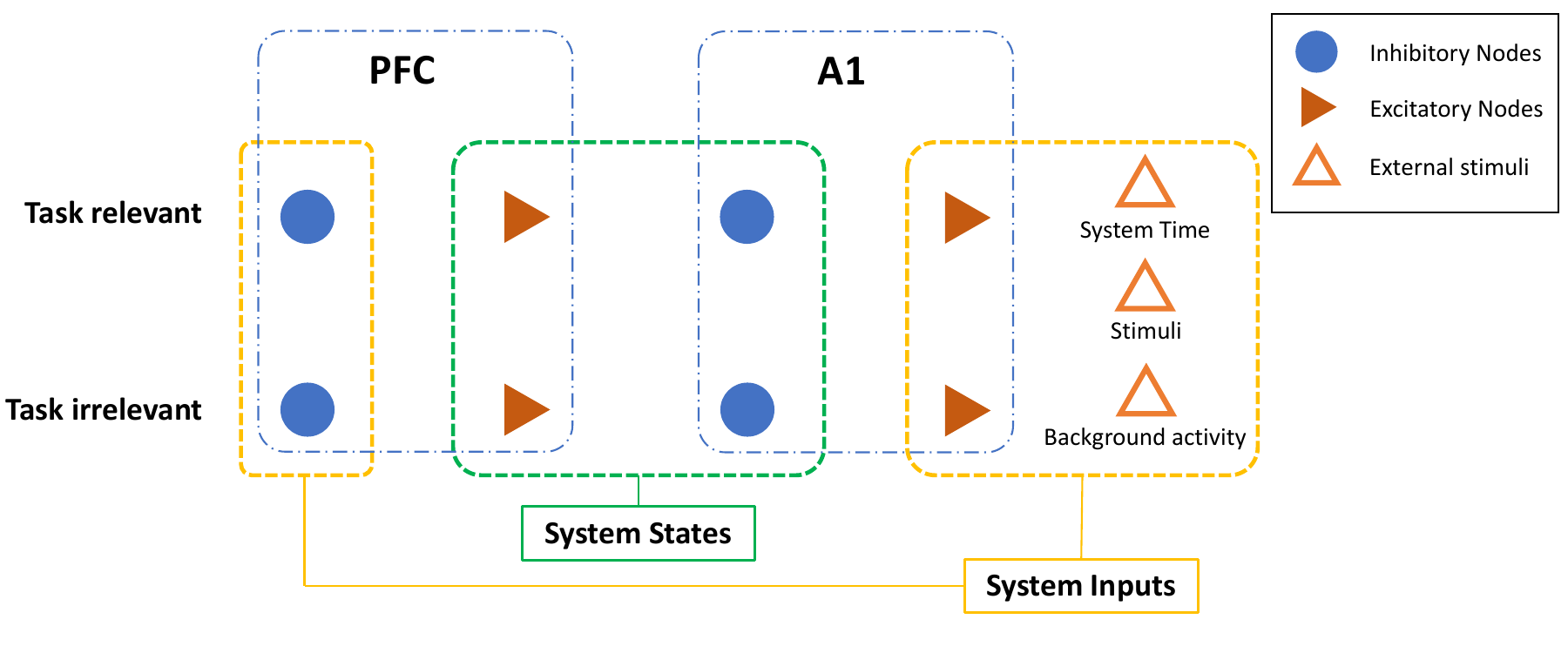}
  \caption{The model includes 8 groups of neuron cells and 3 external
    stimuli. The green box gives the system states; and the yellow
    boxes give the system inputs.}
  \label{Fig_neurons}
\end{figure}
Then, we choose 4 nodes as the system states $(n=4)$, which correspond
to the A1-inhibitory-relevant (A1-IH-TR); the A1-inhibitory-irrelevant
(A1-IH-TI); the PFC-excitatory-relevant (PFC-EX-TR); and
PFC-excitatory-irrelevant (PFC-EX-TI) groups of neurons. These neurons
share similar time constants
\cite{JDM-AB-DJF-RR-JDW-XC-CP-TP-HS-DL-XW:14}.  Finally, we take the
readings of the other 4 nodes, along with three extra signals (i.e.,
system time $u_t=t$, impulse stimuli $u_s=\Delta(t)$, and a constant
background activity $u_b=1$) as system inputs. Thus, the dimension of
the input is $m=7$.

\begin{figure}[h!]
  \vspace{-0.1cm}
  \centering
  \includegraphics[width=8 cm]{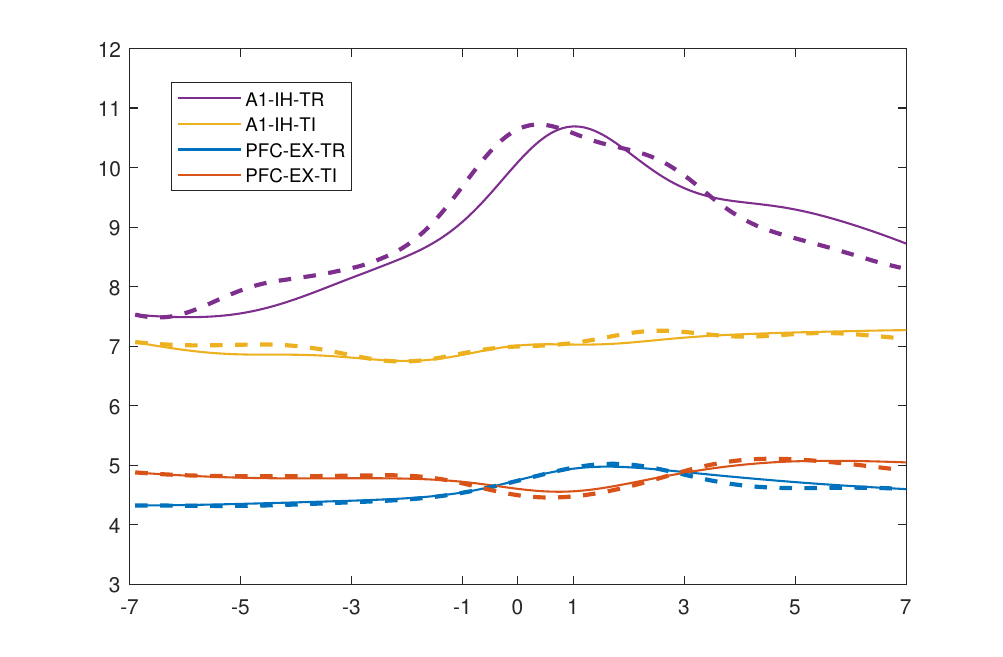}
  \caption{Reconstructing the firing rate dynamics in rodents'
    brain\cite{CCR-MRD:14-crcns}. The solid lines are the experimental
    data; the dashed lines are the dynamics reconstructed by system
    model \eqref{eq_dstmodel}.}
  \label{Fig_rodents}
\end{figure}

Given the network model and the given data set, we employ Algorithm
\ref{Algorithm_Main2} to reconstruct the firing rate dynamics. Here,
we slightly modify the definitions of $\mathcal{P}$ and $\bm{h}$ to
allow the excitatory nodes to have self-loops\footnote{In model
  \eqref{eq_dstmodel}, we assumed that the diagonal entries of $W_D$
  are zero. In the simulation we now allow the two entries to be
  positive, which are corresponding to the two excitatory PFC
  nodes. This modification causes minimal impact to Theorem
  \ref{TH_alg2}.}.  After obtaining the identified the system
parameters, we use the same initial state at $x(t=-7)$, to compare the
experimental data and the firing rate dynamics reconstructed by our
model in Fig.~\ref{Fig_rodents}.  One can see that the identified
linear-threshold network model is able to capture the trends of the
real experimental data.

Apart from the consistency between the reconstructed dynamics and the
data, we also highlight some of our observations from the identified
matrices $W_D$ and $B_D$.  Note that due to the discretization of the
system model and the normalization of data, it is not the absolute
values but the relative ones in $W_D$ and $B_D$ that we should
consider.  For the system matrix,
\begin{itemize}
\item stronger connections ($0.03\sim0.1$) are observed from the two
  PFC-EX (excitatory) nodes to the two A1-IH (inhibitory) nodes,
  regardless of their relevance to the task;
\item In contrast, the connections from the A1-IH nodes to the PFC-EX
  nodes are weak ($\le0.012$);
\item A stronger edge from the PFC-EX-TR (task relevant) to the
  PFC-EX-IR (irrelevant) is observed ($\approx0.04$), but not vice
  versa ($\le0.01$);
\item The connections between the two A1-IH nodes are small
  ($\le0.005$).
\end{itemize}
These observations are in agreement with the hierarchical structure in
selective listening~\cite{EN-JC:21-tacII}, where the level of PFC is
higher than that of A1. Thus, their activities show a single direction
of impact, from PFC-EX neurons to A1-IH neurons.

For the input matrix, major input signals includes stimuli, the
activities of PFC-IH nodes and A1-EX nodes,
\begin{itemize}
\item The stimuli input has stronger impact to the task-relevant nodes
  in both the PFC and A1 areas;
\item The PFC-IH nodes have a larger impact to the PFC-EX nodes;
\item The A1-EX nodes have a larger impact to the A1-IH nodes.
\end{itemize}
These observations show that task-relevant nodes are more sensitive to
the stimuli, which is consistent with the way they are classified, cf.
Fig.~\ref{Fig_neurons}. The results also show that neurons in the same
areas have stronger interactions, which is anatomically reasonable.

%

%

%
%


\section{Conclusions}
We have introduced computationally efficient algorithms to reconstruct
aggregate firing rate dynamics of brain neural networks modeled by
linear-threshold networks.  Central to our approach is a two-step
identification process for model parameters, where we first
reformulate the problem into a scalar variable optimization and then
compute other variables based on the identified scalar variable. Such
decomposition significantly improves computational efficiency, with
guaranteed correctness of the identification results. We have also
considered the impact of measurement noise and proposed a modified
version of the algorithm whose identification error is guaranteed to
be linearly bounded by the magnitude of the error. We have validated
the effectiveness of both algorithms in simulation and on experimental
data.  Future work will leverage the results of the paper in the
design of schemes for the data-driven regulation of neuronal firing
activity, explore the implications for the treatment of brain
disorders, and use real-time data to predict the firing patterns of
animal subjects and its relationship with various cognitive
processeses.

\section*{Appendix}

\begin{proof}[Proof of Lemma~\ref{LM_choiceEC}]
  Due to the one-to-one correspondence (bijection) between $E(\alpha)$
  and $C(\alpha)$, in the following, we only prove the statements of
  Lemma \ref{LM_choiceEC} for $E(\alpha)$.  Since
  $\maxv(\mathcal{X}^+-\alpha \mathcal{X})$ is a piece-wise linear
  function to $\alpha$, it is easy to observe that $E(\alpha)$ changes
  piece-wisely with $\alpha$. Furthermore, from the definition of
  $E(\alpha)$ in equations \eqref{eq_defcE}, we can see that its
  entries ($=0,\pm1$) depend on three types of data in
  $\left(\mathcal{X}^+ - \alpha\mathcal{X}\right)$: the ones active
  the upper saturation threshold
  ($\maxv\left(\mathcal{X}^+ - \alpha\mathcal{X}\right)$); the ones do
  not active saturation thresholds; and the ones that active the lower
  saturation threshold ($0$).  For the convenience of presentation, we
  use three sets $\mathcal{S}$, $\mathcal{M}$, $\mathcal{Z}$, to
  denote these entries ($1$, $0$, $-1$), respectively. Note that the
  matrix $E(\alpha)$ changes, only if the sets $\mathcal{S}$,
  $\mathcal{M}$, $\mathcal{Z}$ change their elements. In order to
  detect such changes, we gradually change the value of $\alpha$ from
  $0$ to $1$, and introduce several \textit{markers} on this range,
  denoted by $\psi_{\ell}$, $\ell=1,2,\cdots$. These $\psi_{\ell}$ are
  located at the transition points of $\alpha$, such that when
  $\alpha$ moves across $\psi_{\ell}$, certain entries of the vector
  $\left(\mathcal{X}^+ - \alpha\mathcal{X}\right)$ will shift from set
  $\mathcal{M}$ to sets $\mathcal{S}$ or $\mathcal{Z}$
  (correspondingly, some other entries of the vector will leave the
  sets $\mathcal{S}$ or $\mathcal{Z}$ and join set
  $\mathcal{M}$). Based on these markers $\psi_{\ell}$, we can
  partition the feasible region of $\alpha$ in to a finite number of
  segments, and we know that when $\alpha$ is on certain segment,
  i.e., $(\psi_{\ell},\psi_{\ell+1})$, the matrix $C(\alpha)$ does not
  change.

  Now to check the total number of possible $\psi_{\ell}$ one can
  create on $(0,1)$, we first consider the element exchange between
  $\mathcal{S}$ and $\mathcal{M}$. Note that
  $\maxv(\mathcal{X}^+-\alpha \mathcal{X})=\maxv(\mathcal{X}^+[i]-\alpha
  \mathcal{X}[i])$, $i=1,\cdots,nT_d$.  Since
  $\maxv(\mathcal{X}^+[i]-\alpha \mathcal{X}[i])$ is convex on
  $\alpha$, for each $i$, the line
  $\mathcal{X}^+[i]-\alpha \mathcal{X}[i]$ can intersect
  $\maxv(\mathcal{X}^+[i]-\alpha \mathcal{X}[i])$ only once, either for
  a continuous interval of $\alpha$ or on a particular point. Thus, on
  different intervals $\alpha\in(\psi_{\ell},\psi_{\ell+1})$, the
  corresponding $\mathcal{S}$ sets must be disjoint. Since the vector
  $(\mathcal{X}^+-\alpha \mathcal{X})$ has $nT_d$ entries, which can
  be partitioned into at most $nT_d$ disjoint sets, the change on
  $\mathcal{S}$ can lead to a at most $nT_d$ number of
  $\psi_{\ell}$. Similarly, for the element exchange between
  $\mathcal{Z}$ and $\mathcal{M}$, since each
  $\mathcal{X}^+[i]-\alpha \mathcal{X}[i]$ can only intersect $0$ for
  one time, it can also create a at most $nT_d$ number of
  $\psi_{\ell}$. Bringing these two conditions together, one has at
  most $2nT_d$ number of $\psi_{\ell}$ on $(0,1)$.

  Finally, for each open sets $(\psi_{\ell},\psi_{\ell+1})$, we have a
  fixed $E(\alpha)$. If we take $\psi_0=1$ and $\psi_{\max}=1$, the
  intervals between $\psi_0, \psi_1, \psi_2,\cdots,\psi_{\max}$ can
  lead to at most $2nT_d+1$ number of $E(\alpha)$. However, it is
  worth mentioning that the union of these open sets
  $(\psi_{\ell},\psi_{\ell+1})$ does not include the marker points
  $\psi_{\ell}$, $\ell=1,2,\cdots$. Actually, on these points, certain
  entries of the vectors
  $\left(\mathcal{X}^+ - \alpha\mathcal{X}\right)$ are intersecting,
  and take the greatest/zero values simultaneously. Thus, the
  $E(\alpha=\psi_{\ell})$ will be different from both
  $E(\alpha<\psi_{\ell})$ and $E(\alpha>\psi_{\ell})$. Considering
  this fact, we will have an extra $2nT_d$ number of $E(\alpha)$ on
  these $\psi_{\ell}$ points. This, together with the previous
  $2nT_d+1$, lead to at most $4nT_d+1$ number of $E(\alpha)$. This
  completes the proof.
\end{proof}

\bibliographystyle{IEEEtran}
\bibliography{alias,New,Main-add,JC,Main}

\begin{IEEEbiography}[{\includegraphics[width=1.0
    in,height=1.35in,clip,keepaspectratio]{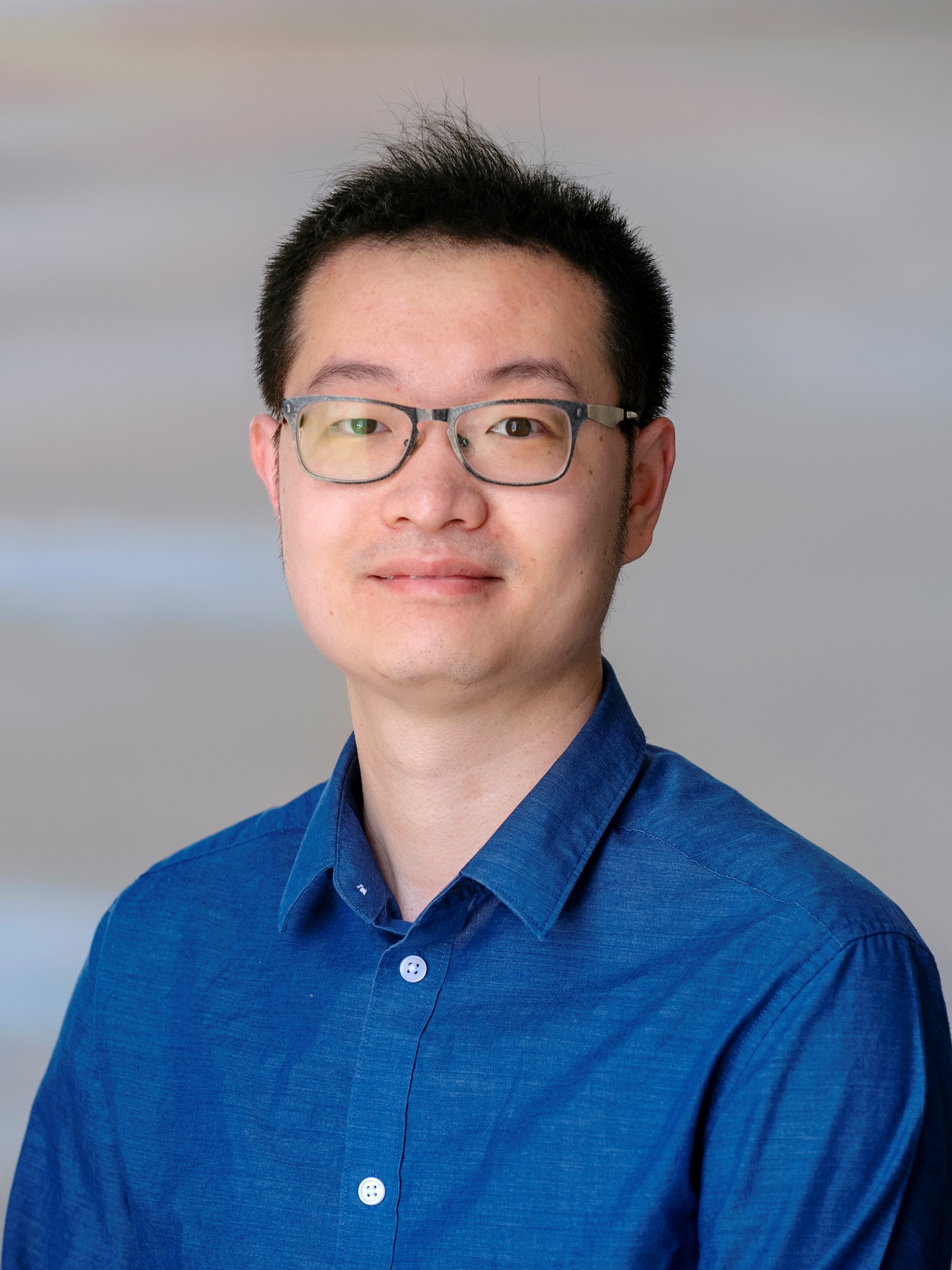}}]{Xuan Wang}{\space}
  is an assistant professor with the Department of Electrical and
  Computer Engineering at George Mason University.  He received his
  Ph.D. degree in autonomy and control, from the School of Aeronautics
  and Astronautics, Purdue University in 2020. He was a post-doctoral
  researcher with the Department of Mechanical and Aerospace
  Engineering at the University of California, San Diego from 2020 to
  2021. His research interests include multi-agent control and
  optimization; resilient multi-agent coordination; system
  identification and data-driven control of network dynamical systems.
\end{IEEEbiography}

\begin{IEEEbiography}[{\includegraphics[width=1.0
    in,height=1.35in,clip,keepaspectratio]{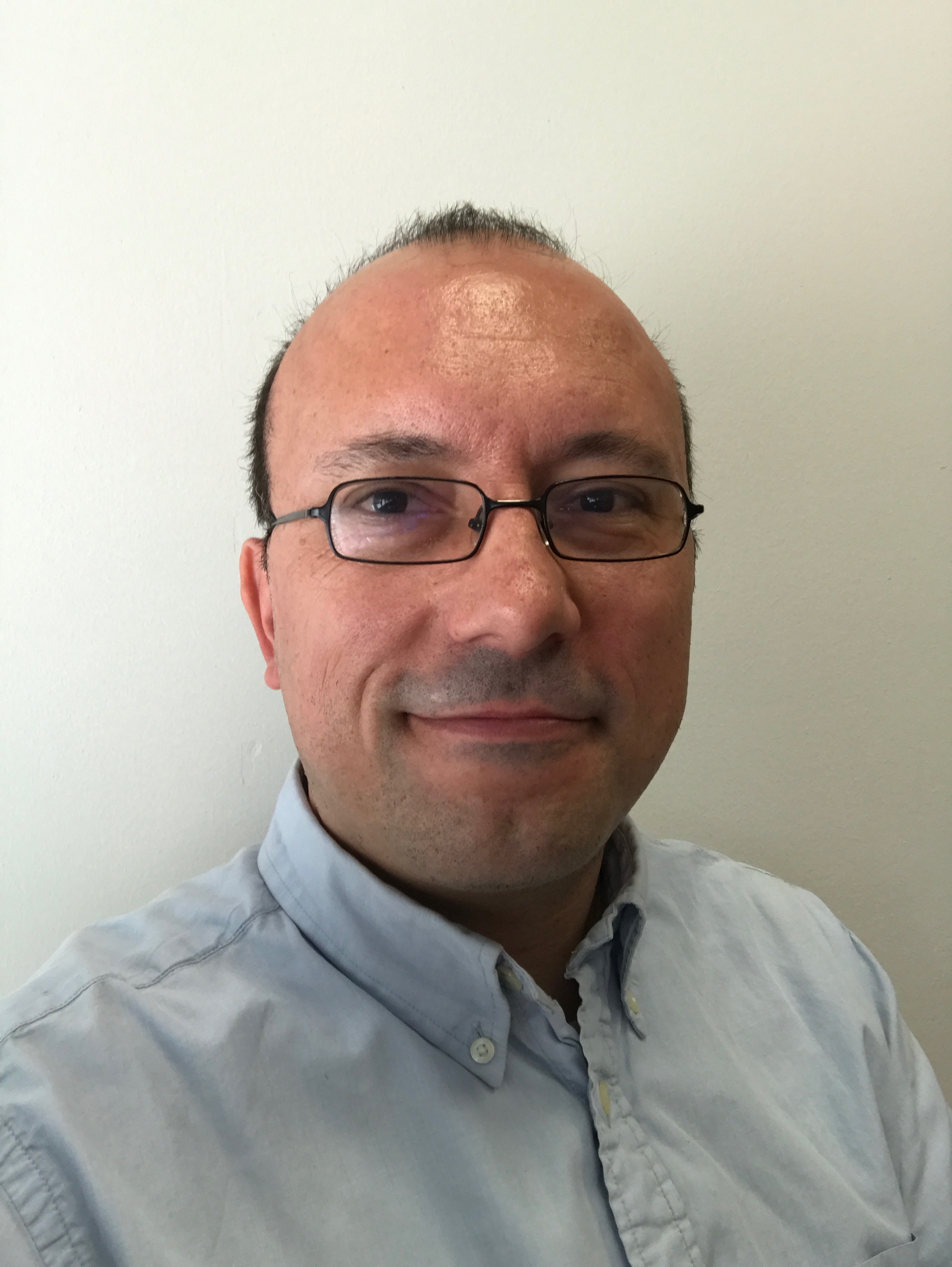}}]{Jorge Cort\'es}{\space}
  received the Licenciatura degree in mathematics from Universidad de
  Zaragoza, Zaragoza, Spain, in 1997, and the Ph.D. degree in
  engineering mathematics from Universidad Carlos III de Madrid,
  Madrid, Spain, in 2001. He held postdoctoral positions with the
  University of Twente, Twente, The Netherlands, and the University of
  Illinois at Urbana-Champaign, Urbana, IL, USA. He was an Assistant
  Professor with the Department of Applied Mathematics and Statistics,
  University of California, Santa Cruz, CA, USA, from 2004 to 2007. He
  is a Professor in the Department of Mechanical and Aerospace
  Engineering, University of California, San Diego, CA, USA.  He is a
  Fellow of IEEE, SIAM, and IFAC.  His research interests include
  distributed control and optimization, network science, non-smooth
  analysis, reasoning and decision making under uncertainty, network
  neuroscience, and multi-agent coordination in robotic, power, and
  transportation networks.
\end{IEEEbiography}

\end{document}